\documentclass[11pt,twoside]{article}
\pdfoutput=1
\usepackage{RR}
\usepackage{geometry}
\usepackage[english]{babel}
\usepackage[utf8]{inputenc}
\usepackage{amssymb}
\usepackage{mathtools}
\usepackage{graphicx}
\usepackage[colorinlistoftodos]{todonotes}
\usepackage{hyperref}
\usepackage{url}
\usepackage{prettyref}
\usepackage{comment}
\usepackage{cancel}
\usepackage{pifont}
\usepackage{algorithm}
\usepackage[noend]{algpseudocode}
\usepackage{manfnt}
\usepackage{microtype}
\hypersetup{
    colorlinks = true,
    urlbordercolor = {magenta},
    citecolor = {blue},
}

\usepackage{todonotes}

\definecolor{modelica-model}{HTML}{1E7A19}
\definecolor{modelica-bool}{HTML}{0F7877}
\definecolor{modelica-eq-ite}{HTML}{7F007E}
\definecolor{modelica-der}{HTML}{0000FC}

\usepackage{color}
\definecolor{darkred}{rgb}{0.6,0,0}
\definecolor{darkorange}{rgb}{0.7,0.2,0}
\definecolor{darkgreen}{rgb}{0,0.3,0}
\definecolor{lightblue}{rgb}{0.3,0.3,1}

\usepackage{xcolor}
\definecolor{mygreen}{rgb}{0.0, 0.42, 0.24}
\definecolor{myblue}{rgb}{0.16, 0.32, 0.75}
\definecolor{mychoc}{rgb}{0.48, 0.25, 0.0}

\newcommand{\sigmamethod}{$\Sigma$-method}

\newcommand{\proj}[2]{\mathbf{proj}_{#1}\!\left(#2\right)}
\newcommand{\rref}[2][]{\prettyref{#2}}
\newrefformat{sec}{Section\,\ref{#1}}
\newrefformat{def}{Definition\,\ref{#1}}
\newrefformat{defn}{Definition\,\ref{#1}}
\newrefformat{thm}{Theorem\,\ref{#1}}
\newrefformat{prop}{Prop.\,\ref{#1}}
\newrefformat{lem}{Lem.\,\ref{#1}}
\newrefformat{cor}{Corollary\,\ref{#1}}
\newrefformat{ex}{Ex.\,\ref{#1}}
\newrefformat{ax}{Appendix\,\ref{#1}}
\newrefformat{tab}{Table\,\ref{#1}}
\newrefformat{fig}{Figure\,\ref{#1}}
\newrefformat{alg}{Algorithm\,\ref{#1}}
\newrefformat{exec}{Exec.\,Sch.\,\ref{#1}}
\newrefformat{op}{Line\,\ref{#1}}
\newrefformat{rk}{Remark\,\ref{#1}}
\newrefformat{pr}{Principle\,\ref{#1}}
\newrefformat{eq}{Eq.\,\eqref{#1}}
\newrefformat{sys}{System\,\eqref{#1}}

\newtheorem{theorem}{Theorem}
\newtheorem{lemma}[theorem]{Lemma}

\newtheorem{notation}[theorem]{Notations}

\newtheorem{problem}{Problem}
\newtheorem{question}{Question}
\newtheorem{tool}{Basic tool}

\newtheorem{definition}[theorem]{Definition}

\newtheorem{corollary}[theorem]{Corollary}

\newtheorem{ccomment}[theorem]{Comment}
\newtheorem{principle}{Principle}

\newenvironment{proof}{\paragraph{\it Proof}}{\eproof \\

}

\def\Jacobian{\mathbf{J}}

\def\bfc{\mathbf{c}}
\def\bfd{\mathbf{d}}

\def\bff{\mathbf{f}}

\def\bfB{\mathbf{B}}

\def\path{\pi}

\def\DM{{\rm DM}}

\def\IsamDAE{IsamDAE}

\usepackage{graphicx}
\graphicspath{{./graphics/}}

\newcommand{\prog}[1]{\mathtt{#1}}

\newcommand{\remph}[1]{\mbox{{\color{red}$#1$}}}

\newcommand{\bemph}[1]{\mbox{{\color{blue}$#1$}}}
\newcommand{\blemph}[1]{\mbox{{\color{black}$#1$}}}

\newcommand{\consistency}[1]{\overline{#1}}

\renewcommand{\dot}[1]{#1'}

\usepackage[english]{babel}

\newcommand{\myparagraph}[1]{\smallskip\noindent\textit{{#1:}}}
\newcommand{\bea}{\begin{array}}
\newcommand{\eea}{\end{array}}
\newcommand{\beq}{\begin{eqnarray}}
\newcommand{\eeq}{\end{eqnarray}}
\newcommand{\beqq}{\begin{eqnarray*}}
\newcommand{\eeqq}{\end{eqnarray*}}

\newcommand{\Ra}{\Rightarrow}
\newcommand{\ra}{\rightarrow}
\newcommand{\la}{\leftarrow}

\newcommand{\cJ}{\mathcal{J}}
\newcommand{\val}{\nu}
\newcommand{\vval}[1]{\val_{#1}}

\newcommand{\cC}{\mathcal{C}}

\newcommand{\cG}{\mathcal{G}}
\newcommand{\Edges}[1]{E_{#1}}

\newcommand{\eqq}{e}

\newcommand{\block}{\beta}

\newcommand{\preset}[1]{{^{\bullet\!}{#1}}}

\newcommand{\postset}[1]{#1^\bullet}

\newcommand{\ppostset}[2]{#2^{\bullet#1}}
\newcommand{\pprime}[2]{#2^{\prime#1}}

\newcommand{\eproof}{\hfill$\Box$}
\newcommand{\vsmall}{\partial}

\newcommand{\regularU}[2]{U_{\!#1\!}\left(#2\right)}
\newcommand{\regularV}[2]{V_{\!#1\!}\left(#2\right)}

\newcommand{\eqdef}{=_{{\rm def}}}

\newcommand{\guard}{\gamma}

\newcommand{\squared}{\mathfrak{E}}
\newcommand{\underapprox}{\mathfrak{U}}
\newcommand{\overapprox}{\mathfrak{O}}

\newcommand{\undefined}[1]{\mathit{Undef}(#1)}

\newcommand{\vvec}[1]{\mathcal{#1}}

\newcommand{\cA}{{\cal A}}
\newcommand{\cM}{{\cal M}}

\newcommand{\cS}{{\cal S}}
\newcommand{\CS}{\mathbf{ER}}

\newcommand{\cP}{{\cal P}}
\newcommand{\cH}{{\cal H}}

\newcommand{\when}{\prog{if}}
\newcommand{\doo}{\prog{then}}
\newcommand{\mdAE}{mdAE}

\newcommand{\mDAE}{mDAE}

\newcommand{\bR}{\mathbb{R}}
\newcommand{\bT}{\mathbb{T}}

\newcommand{\bfx}{\mathbf{x}}

\newcommand{\bN}{\mathbb{N}}

\newcommand{\ttt}{\mbox{\sc t}}
\newcommand{\fff}{\mbox{\sc f}}

\newcommand{\aand}{\mathtt{and}}

\newcommand{\atomicact}[1]{\mathsf{#1}}

\newcommand{\REMOVE}[1]{}

\RRdate{January 2021}

\RRetitle{
\bf Structural Analysis of Multimode DAE Systems: summary of results
} 
\RRtitle{
\bf Analyse Structurelle des Syst\`emes de DAE multimodes: \\ recueil des r\'esultats de base
} 

\RRauthor{\rm Albert Benveniste\thanks{All authors are with Inria-Rennes, Campus de Beaulieu, 35042 Rennes cedex, France; email: surname.name@inria.fr}, Benoit Caillaud, 
Mathias Malandain
}
\authorhead{A. Benveniste, B. Caillaud, M. Malandain}

\vfill
\RRnote{This work was supported by the FUI ModeliScale DOS0066450/00 French national grant (2018-2021) and the Inria IPL ModeliScale large scale initiative (2017-2021, \url{https://team.inria.fr/modeliscale/}).}

\RRabstract{Modern modeling languages for general physical systems, such as Modelica, Amesim, or Simscape, rely on Differential Algebraic Equations (DAEs), i.e., constraints of the form $f(x',x,u)=0$. This drastically facilitates modeling from first principles of the physics, as well as model reuse. In recent works~\cite{BENVENISTE2020,DBLP:journals/corr/abs-2008-05166}, we presented the mathematical theory needed to establish the development of compilers and tools for DAE-based physical modeling languages on solid mathematical grounds.

\hspace*{1em} At the core of this analysis sits the so-called \emph{structural analysis}, whose purpose, at compile time, is to either identify under- and overspecified subsystems (if any), or to rewrite the model in a form amenable of existing DAE solvers, including the handling of mode change events. The notion of ``structure'' collects, for each mode and mode change event, the variables and equations involved, as well as the \emph{latent equations} (additional equations redundant with the system), needed to prepare the code submitted to the solver. The notion of DAE \emph{index} (the minimal number of times any equation has to be possibly differentiated) is part of this structural analysis. 

\hspace*{1em} This report complements~\cite{BENVENISTE2020,DBLP:journals/corr/abs-2008-05166} by collecting all the needed background on structural analysis. 
The body of knowledge on structural analysis is large and scattered, which also motivated us to collect it in a single report.

\hspace*{1em} We first explain the primary meaning of structural analysis of systems of equations, namely the study of their regularity or singularity in some generic sense. We then briefly review the body of graph theory used in this context. We develop some extensions, for which we are not aware of any reference, namely the structural analysis of systems of equations with existential quantifiers. 

\hspace*{1em} For the structural analysis of DAE systems, we focus on John Pryce's $\Sigma$-method, that we both summarize and extend to non-square systems.

\hspace*{1em} The uses of these tools and methods in~\cite{BENVENISTE2020,DBLP:journals/corr/abs-2008-05166} are highlighted in this report.
}

\RRresume{
Les langages modernes de modélisation de systèmes physiques, tels que Modelica, Amesim ou Simscape, s'appuient sur des Équations Algébro-Différentielles (DAE, de l'anglais \textit{Differential-Algebraic Equations}), c'est-à-dire des contraintes de la forme $f(x',x,u)=0$.
La modélisation à partir des premiers principes de la physique, ainsi que la réutilisation de modèles, sont facilitées par cette approche.
Dans des travaux récents~\cite{BENVENISTE2020,DBLP:journals/corr/abs-2008-05166}, nous présentons la théorie mathématique requise pour le développement de compilateurs et d'outils pour les langages de modélisation à base de DAE sur des bases mathématiques solides.

Cette analyse s'appuie sur l'\textit{analyse structurelle}, dont le but, à la compilation d'un modèle, est soit d'identifier d'éventuels sous-systèmes sous- et sur-déterminés, soit de réécrire le modèle en vue de son traitement par des solveurs de DAE existants, en prenant en compte les événements de changements de mode.
La notion de "structure" rassemble, pour chaque mode et chaque changement de mode, les équations et les variables impliquées, ainsi que les \textit{équations latentes} (des équations supplémentaires redondantes), requises pour construire le code soumis au solveur numérique.
La notion d'\textit{index} d'un système de DAE, lié aux nombres de différentiations successives devant être appliquées à ses équations, est partie intégrante de cette analyse structurelle.

Le présent rapport vient en complément de~\cite{BENVENISTE2020,DBLP:journals/corr/abs-2008-05166}, en rassemblant toutes les connaissances requises en analyse structurelle.
Ce travail a également été motivé par le fait que ce vaste corpus de connaissances est éparpillé dans la littérature.

Nous expliquons d'abord le sens premier de l'analyse structurelle de systèmes d'équations, à savoir, l'étude de leur régularité ou singularité dans un sens générique.
Nous passons ensuite en revue les éléments de théorie des graphes utilisés dans ce contexte.
Nous développons également une extension qui est, à notre connaissance, inédite dans la littérature sur le sujet : l'analyse structurelle de systèmes d'équations contenant des quantificateurs existentiels.

Pour l'analyse structurelle de systèmes de DAE, nous nous focalisons sur la $\Sigma$-méthode de J. Pryce, qui est résumée puis étendue à des systèmes non-carrés.

Les utilisations de ces outils et méthodes dans~\cite{BENVENISTE2020,DBLP:journals/corr/abs-2008-05166} sont mises en évidence dans ce rapport.
}

\RRmotcle{analyse structurelle, \'equations algébro-diff\'erentielles
  (DAE), syst\`emes multi-mode, mod\`eles \`a structure variable}
\RRkeyword{structural analysis, differential-algebraic equations
  (DAE), multi-mode systems, variable-structure models}
	\RRprojet{Hycomes}  % cas d'un seul projet
\RCRennes
\begin{document}
\RRNo{9387}
\hypersetup{pageanchor=false}
\begin{titlepage}
\makeRR
\end{titlepage}
\hypersetup{pageanchor=true}
\pagenumbering{arabic}
\clearpage
\tableofcontents
\clearpage
\section{{Introduction}}\label{sec:intro}
Modern modeling languages for general physical systems, such as Modelica, Amesim, or Simscape, rely on Differential Algebraic Equations (DAEs), i.e., constraints of the form $f(\dot{x},x,u)=0$. This drastically facilitates modeling from first principles of the physics, as well as model reuse. In recent works~\cite{BENVENISTE2020,DBLP:journals/corr/abs-2008-05166}, we presented the mathematical theory needed to establish the development of compilers and tools for DAE-based physical modeling languages on solid mathematical grounds.

%%%%%%%%%%%%%%%%%%%%%%%%%%%%%%
At the core of this analysis sits the so-called \emph{structural analysis}, whose purpose, at compile time, is to either identify under- and overspecified subsystems (if any), or to rewrite the model in a form amenable of existing DAE solvers, including the handling of mode change events. The notion of ``structure'' collects, for each mode and mode change event, the variables and equations involved, as well as the \emph{latent equations} (additional equations redundant with the system), needed to prepare the code submitted to the solver. The notion of DAE \emph{index} (the minimal number of times any equation has to be possibly differentiated) is part of this structural analysis. The body of knowledge on structural analysis is large and scattered, which motivated us to collect it in a single report.  
 This report complements~\cite{BENVENISTE2020,DBLP:journals/corr/abs-2008-05166} by collecting all the needed background on structural analysis. The uses of these tools and methods in~\cite{BENVENISTE2020,DBLP:journals/corr/abs-2008-05166} are highlighted in this report.

\myparagraph{Structural analysis of systems of equations}
We first explain the primary meaning of the structural analysis of systems of equations, namely the study of their regularity or singularity in some generic sense.
%% Old version
\begin{comment}
Sparse matrices are matrices in which non-zero entries are seldom in proportion. The association of graphs with sparse matrices is pervasive in linear algebra for high performance computing~\cite{Hogb06,hogben14}. 

The most commonly used tool is the bipartite graph characterizing the pattern of the matrix (the $(i,j)$-location of its non-zero entries). If, by permuting rows and columns, a square matrix can be put in a form with non-zero elements in its diagonal, then we say that this matrix is structurally regular; this notion is repeatedly used in the Modelica community. Structural regularity or singularity and derived notions, are checked on the bipartite graph associated to a matrix.
\end{comment}
%% New version (longer)
Structural information about a system of equations can be summed up by its incidence matrix; in particular, a square system yields a square matrix.
If, after some permutation of row and columns, this matrix only has non-zero diagonal entries, then we say that this matrix, as well as the underlying system of equations, are structurally regular: this notion is repeatedly used in the Modelica community.

However, the most commonly used structural representation of a system of equations is actually the bipartite graph characterizing the pattern of the incidence matrix (i.e., the $(i,j)$-locations of its non-zero entries).
Structural regularity or singularity, as well as derived notions, are then checked on this bipartite graph.
As a matter of fact, the systems of equations under study are often sparse (i.e., their incidence matrices have a small proportion of non-zero entries), and the association of graphs with sparse matrices is pervasive in linear algebra for high performance computing~\cite{Hogb06,hogben14}. 
%% End of new version
We briefly review the body of graph theory used in this context. We recall in particular the Dulmage-Mendelsohn decomposition of a bipartite graph, which is the basis for structuring a general (non-square) matrix into its overdetermined, regular, and underdetermined blocks, structurally. This body of knowledge in linear algebra is referred to as \emph{structural analysis} of matrices. Structural analysis extends to systems of numerical equations, by considering the Jacobian of the considered system.

The link between numerical regularity and structural regularity belongs to the folklore of linear algebra; however, its understanding is critical for a correct use of structural analysis for systems of numerical equations. We were, however, unable to find a reference where related results were detailed; we propose such results here.

We also develop some useful extensions, for which we are not aware of any reference, particularly the structural analysis of systems of equations with existential quantifiers. 

\myparagraph{Structural analysis of DAE systems}
We focus on DAE and summarize John Pryce's \sigmamethod\ for structural analysis~\cite{Pryce01}, and we present an extension of this method to non-square DAE systems. We also include the structural form of Campbell and Gear's method of \emph{differential arrays}~\cite{CampbellGear1995}, which we complement with its discrete-time counterpart, namely ``difference arrays''. These methods rely on the structural analysis of systems of equations with existential quantifiers.

%\paragraph{Other uses of structural analysis exist than generating simulation code:}  We complement this background material on structural analysis of multimode DAE systems, by informally discussing other uses of structural analysis than just generating simulations (requirement engineering, diagnosis and health monitoring, and more). We explain why such alternative use cases require extensions of structural analysis algorithms. This draws a roadmap for future work.

\section{Use of structural analysis on a toy example}
This section is a summary of Sections~1.2.1 and~2 of~\cite{DBLP:journals/corr/abs-2008-05166}, to which the reader is referred for further details.

\subsection{An ideal clutch}
Our example is an idealized clutch involving two rotating shafts with no motor or brake connected (\rref{fig:clutch}). We assume this system to be closed, with no interaction other than explicitly specified.
\begin{figure}[h]
  \centering
    \includegraphics[height=2.5cm]{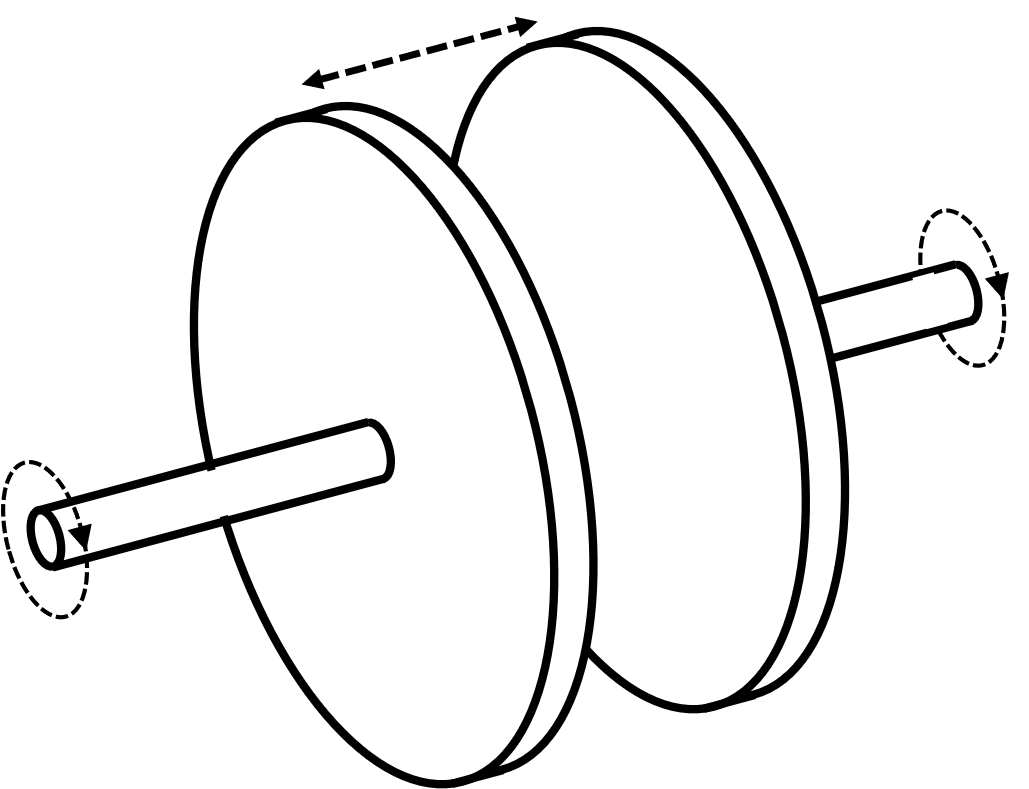}
  \caption{An ideal clutch with two shafts.}
 \label{fig:clutch}
\end{figure}

The dynamics of each shaft $i=1,2$, is modeled by \mbox{$\omega_i'=f_i(\omega_i,\tau_i)$} for some, yet unspecified, function $f_i$, where $\omega_i$ is the angular velocity, $\tau_i$ is the torque applied to shaft $i$, and $\omega_i'$ denotes the time derivative of $\omega_i$. Depending on the value of the input Boolean variable $\gamma$, the clutch is either engaged ($\gamma = \ttt$, the constant ``true'') or released ($\gamma = \fff$, the constant ``false''). When the clutch is released, the two shafts rotate freely: no torque is applied to them ($\tau_i=0$).  When the clutch is engaged, it ensures a perfect join between the two shafts, forcing them to have the same angular velocity ($\omega_1-\omega_2=0$) and opposite torques ($\tau_1+\tau_2=0$). Here is the model:
\begin{equation}
\left\{
\bea{rllcc}
&&\omega'_1=f_1(\omega_1,\tau_1) &&(\eqq_1) \\
&&\omega'_2=f_2(\omega_2,\tau_2) &&(\eqq_2) \\
\when\;\guard&\doo&\omega_1-\omega_2=0 &&(\eqq_{3}) \\
&\aand&\tau_1+\tau_2=0 &&(\eqq_{4}) \\
\when\;\prog{not}\;\guard&\doo&\tau_1=0 &&(\eqq_{5}) \\
&\aand&\tau_2=0 &&(\eqq_{6}) \\
\eea
\right.
\label{sys:coupledshafts}
\end{equation}
When $\guard=\ttt$, equations $(\eqq_3,\eqq_4)$ are active and equations $(\eqq_5,\eqq_6)$ are disabled, and vice-versa when $\guard=\fff$. If the clutch is initially released, then, at the instant of contact, the relative speed of the two rotating shafts jumps to zero; as a consequence, an impulse is expected on the torques. 

\subsection{Analyzing the two modes separately}
This model yields an ODE system when the clutch is released ($\guard=\fff$), and a DAE system of index $1$ when the clutch is engaged ($\guard=\ttt$), due to equation $(\eqq_3)$ being active in this mode.
The basic way of generating simulation code for the engaged mode ($\guard=\ttt$) is to add, to the model in this mode, the \emph{latent equation} $\remph{(\eqq'_{3})}$:
\begin{equation}
\left\{
\bea{lcc}
\omega'_1=f_1(\omega_1,\tau_1) &&(\eqq_1) \\
\omega'_2=f_2(\omega_2,\tau_2) &&(\eqq_2) \\
\omega_1-\omega_2=0 &&(\eqq_{3}) \\
\remph{\omega'_1-\omega'_2=0} &&\remph{(\eqq'_{3})} \\
\tau_1+\tau_2=0 &&(\eqq_{4}) 
\eea
\right.
\label{sys:engaged}
\end{equation}
By doing so, we can apply the following basic policy:
\beq
\mbox{
\begin{minipage}{14cm}
	 Regard \rref{sys:engaged} as a system of equations in which the dependent variables are the leading variables of \rref{sys:engaged}, i.e., the torques $\tau_i$ and the accelerations $\omega'_i$, and consider the $\omega'_i$'s as \emph{dummy derivatives,} i.e., fresh variables not related to the velocities $\omega_i$; ignore equation $(\eqq_{3})$ and solve the four remaining equations for the dependent variables, using an algebraic equation solver. 
	%This brings \rref{sys:engaged} to becoming as if it was an ODE.
\end{minipage}
}
\label{rtihjloiusdytfuy}
\eeq
This policy amounts to, first, bringing the DAE system back to an ``ODE-like'' form, then solving it as such. In this case, once the latent equation $(\eqq'_{3})$ was added, the Jacobian matrix of the system $(\eqq_1,\eqq_2,\eqq'_3,\eqq_4)$ with respect to its leading variables was invertible.
%The goal of this approach is that we have added latent equation ${(\eqq'_{3})}$ to ensure that the Jacobian of \rref{sys:engaged} from which $(\eqq_{3})$ is removed, shall be invertible. 
This is the classical approach to DAE simulation via ``index reduction''---having differentiated ${(\eqq_{3})}$ once indicates that \rref{sys:engaged} has index $1$. 
\begin{tool}
	\label{ltrgikuliu} The current practice in DAE-based modeling languages is that adding latent equations is performed on the basis of a \emph{structural analysis,} which is a semi-decision procedure using graph theoretic criteria, of much lower computational cost than that of studying the numerical invertibility of the Jacobian. We provide the bases for structural analysis of DAEs in Sections~$\ref{lrieughltiughliu}$ and~$\ref{lfdkuvekrughluh}$ of this report.
\end{tool}
The analysis of the system above in both its modes is classical. The handling of mode changes, however, is not, due to the occurrence of impulsive torques in the ``released\,$\ra$\,engaged'' transition. Also, the nature of time differs between the study of modes and that of mode changes. Each mode, ``released'' or ``engaged'', evolves in \emph{continuous-time} according to a DAE system. In contrast, the mode changes are events at which restart values for states must be computed from states before the change, in one or several \emph{discrete-time} transitions---note that such restart transitions are not explicitly specified in \rref{sys:coupledshafts}.

\subsection{Analyzing mode changes}
\label{edjywtefjyt}
Let us focus on the ``released\,$\ra$\,engaged'' transition, where torques are impulsive. 

%\subsubsection{Using a discrete-time line}
To unify the views of time, we map the different kinds of time to a same discrete-time line, by discretizing the real time line $\bR_{\geq 0}$ using an \emph{infinitesimal} time step $\vsmall$. By doing so, the DAEs acting within the two ``released'' or ``engaged'' modes, get approximated with an \emph{infinitesimal} error. By ``infinitesimal'', we mean: ``smaller in absolute value than any non-zero real number''. All of this can be mathematically formalized by relying on \emph{Nonstandard Analysis}~\cite{Robinson,Cutland}, see~\cite{BENVENISTE2020,DBLP:journals/corr/abs-2008-05166} for details. In particular, \rref{sys:coupledshafts} is mapped to discrete-time by 
using the discrete-time line
\beq
\bT &\eqdef& \{n\vsmall\mid n=0,1,2,\dots\}
\label{lukgrhlkdut}
\eeq
and replacing in it all the derivatives by their first order explicit Euler schemes, that is:
\beq
%\mbox{in \rref{sys:coupledshafts}, } ~
\dot{\omega_i}(t) \mbox{ is replaced by }\frac{\omega_i(t+\vsmall)-\omega_i(t)}{\vsmall}\,.
\eeq
In (\ref{lukgrhlkdut}), integer $n$ must become large enough so that the whole time line $\bR_{\geq 0}$ is covered---infinite integers must be used, which is also formalized in Nonstandard Analysis. Using the \emph{shift} operators
\beq
\preset{x}(t) \eqdef x(t-\vsmall) &\mbox{and}&
\postset{x}(t) \eqdef x(t+\vsmall)\,,
\label{kbnsdlvbkujvbnhlk}
\eeq
we can rewrite \rref{sys:coupledshafts} as follows:
\begin{equation}
\left\{
\bea{rllcc}
&&\frac{\postset{\omega_1}-\omega_1}{\vsmall}=f_1(\omega_1,\tau_1) &&(\eqq_1) \\
&&\frac{\postset{\omega_2}-\omega_2}{\vsmall}=f_2(\omega_2,\tau_2) &&(\eqq_2) \\
\when\;\guard&\doo&\omega_1-\omega_2=0 &&(\eqq_{3}) \\
&\aand&\tau_1+\tau_2=0 &&(\eqq_{4}) \\
\when\;\prog{not}\;\guard&\doo&\tau_1=0 &&(\eqq_{5}) \\
&\aand&\tau_2=0 &&(\eqq_{6}) \\
\eea
\right.
\label{sys:ledgikujliku}
\end{equation}
Index reduction for the engaged mode can now be performed by shifting $(\eqq_{3})$ and adding it to (\ref{sys:ledgikujliku}), shown in {\color{red}red}:
\begin{equation}
\left\{
\bea{rllcc}
&&\frac{\postset{\omega_1}-\omega_1}{\vsmall}=f_1(\omega_1,\tau_1) &&(\eqq_1) \\
&&\frac{\postset{\omega_2}-\omega_2}{\vsmall}=f_2(\omega_2,\tau_2) &&(\eqq_2) \\
\when\;\guard&\doo&\omega_1-\omega_2=0 &&(\eqq_{3}) \\
&\aand&\remph{\postset{\omega_1}-\postset{\omega_2}=0} &&\remph{(\postset{\eqq_{3}})} \\
&\aand&\tau_1+\tau_2=0 &&(\eqq_{4}) \\
\when\;\prog{not}\;\guard&\doo&\tau_1=0 &&(\eqq_{5}) \\
&\aand&\tau_2=0 &&(\eqq_{6}) \\
\eea
\right.
\label{sys:lsedughlsu}
\end{equation}
\rref{sys:lsedughlsu} will be used to generate restart equations at impulsive mode change $\guard:\fff\ra\ttt$. 
At the
considered instant, we have $\preset{\guard}=\fff$ and
$\guard=\ttt$, see (\ref{kbnsdlvbkujvbnhlk}) for the definition of shift operators. Unfolding
\rref{sys:lsedughlsu} at the two successive previous (shown in
{\color{blue}blue}) and current (shown in black) instants yields, by
taking the actual values for the guard at those instants into account:
\beq\bea{rl} \bemph{\mbox{previous}}&\left\{ \bea{lcl}
  \bemph{\frac{{\omega_1}-\preset{\omega_1}}{\vsmall}=f_1(\preset{\omega_1},\preset{\tau_1})} && \bemph{(\preset{\eqq_1^\vsmall})} \\ [1mm]
  \bemph{\frac{{\omega_2}-\preset{\omega_2}}{\vsmall}=f_2(\preset{\omega_2},\preset{\tau_2})} && \bemph{(\preset{\eqq_2^\vsmall})} \\
  \bemph{\preset{\tau_1}=0} && \\
  \bemph{\preset{\tau_2}=0} && 
\eea
\right.
\\
\blemph{\mbox{current}}&\left\{
\bea{lcl}
\frac{\postset{\omega_1}-\omega_1}{\vsmall}=f_1(\omega_1,\tau_1) &~~\;&  \\ [1mm]
\frac{\postset{\omega_2}-\omega_2}{\vsmall}=f_2(\omega_2,\tau_2) && \\ [1mm]
\omega_1-\omega_2=0 && (\eqq_3) \\ 
\postset{\omega_1}-\postset{\omega_2}=0 && \\ 
\tau_1+\tau_2=0 && 
\eea
\right.
\eea
\label{sys:welguihwtuilh}
\eeq
We regard \rref{sys:welguihwtuilh} as an algebraic system
of equations with dependent variables being the leading variables of \rref{sys:lsedughlsu} at the
previous and current instants: 
$\bemph{\preset{\tau_i},{\omega_i}};{\tau_i},\postset{\omega_i}$ for
$i=1,2$.
\rref{sys:welguihwtuilh} is singular since the following subsystem
possesses five equations and only four dependent variables
$\omega_1,\omega_2,\preset{\tau_1},\preset{\tau_2}$:
\beq \left\{
  \bea{lcl}
  \bemph{\frac{{\omega_1}-\preset{\omega_1}}{\vsmall}=f_1(\preset{\omega_1},\preset{\tau_1})} && \bemph{(\preset{\eqq_1^\vsmall})} \\ [1mm]
  \bemph{\frac{{\omega_2}-\preset{\omega_2}}{\vsmall}=f_2(\preset{\omega_2},\preset{\tau_2})} && \bemph{(\preset{\eqq_2^\vsmall})} \\
  \bemph{\preset{\tau_1}=0} && \\
  \bemph{\preset{\tau_2}=0} && \\
\omega_1-\omega_2=0 && (\eqq_3) 
\eea
\right.
\label{oreiuyoweriu}
\eeq
We propose to resolve the conflict in (\ref{oreiuyoweriu}) by
applying the following causality principle:
\begin{principle}[causality]
	\label{oeuryfgy}  What was done at the previous instant cannot be undone at the current instant.
\end{principle}
This leads to removing, from subsystem (\ref{oreiuyoweriu}), the conflicting equation $(\eqq_3)$, thus getting the following problem for computing restart values at mode change $\guard:\fff\ra\ttt$:
\beq
\mbox{
\begin{minipage}{5.5cm}
	 solve for $\postset{\omega_1},\postset{\omega_2},\tau_1,\tau_2$ the following system, and then, use $\postset{\omega_1},\postset{\omega_2}$  as restart values for the velocities:
\end{minipage}
}
&
%&:&
\left\{
\bea{lcl}\bemph{{\omega_1},{\omega_2},\preset{\tau_1},\preset{\tau_2} \mbox{ set by previous instant}} && \\ [1mm]
\frac{\postset{\omega_1}-\omega_1}{\vsmall}=f_1(\omega_1,\tau_1) &&  \\ [1mm]
\frac{\postset{\omega_2}-\omega_2}{\vsmall}=f_2(\omega_2,\tau_2) && \\ [1mm]
\postset{\omega_1}-\postset{\omega_2}=0 && \\ 
\tau_1+\tau_2=0 && 
\eea
\right.
\label{sys:segfuihpeiu}
\eeq
Note that the consistency equation $(\eqq_3):\omega_1-\omega_2=0$
has been removed from \rref{sys:segfuihpeiu}, thus modifying the original model. However, this removal occurs only at mode
change events \mbox{$\guard:\fff\ra\ttt$}. What we have done amounts
to \emph{delaying by one infinitesimal time step the satisfaction of some
  of the constraints in force in the new mode} $\guard=\ttt$. Since
our time step is infinitesimal, this takes zero real time. The move from \rref{sys:welguihwtuilh} to \rref{sys:segfuihpeiu} must be automatized, hence we need:
\begin{tool}
	\label{helrughsrilu} We need a graph based algorithm for identifying the conflicting equations possibly occurring at a mode change. This is addressed in Section~$\ref{fltorghbfdio}$ of this report, where the Dulmage-Mendelsohn decomposition of a bipartite graph is presented.
\end{tool}
Now, \rref{sys:segfuihpeiu} is not effective yet, since it involves, in the denominator of the left-hand side of the first two equations, the infinitesimal number $\vsmall$. Letting $\vsmall\searrow{0}$ in \rref{sys:segfuihpeiu} requires a different kind of technique, not related to structural analysis, to identify impulsive variables (here, the two torques) and to handle them properly. Details for these techniques are found in Sections 2.5.2 and 10 of~\cite{DBLP:journals/corr/abs-2008-05166}.

\section{Structural vs. numerical analysis of systems of equations}
\label{lrieughltiughliu}
In this section, we develop the structural analysis for a system of algebraic equations, i.e., equations in a Euclidian space, with no dynamics (no time, no derivative), of the form 
\begin{equation}
f_i(y_1,\dotsc,y_k,x_1,\dotsc,x_n)=0,\quad i=1,\dotsc,m\,.
\label{eq:F}
\end{equation}
This system is rewritten as $F(Y,X)=0$, where $Y$ and $X$ denote the vectors
$(y_1,\dotsc,y_k)$ and $(x_1,\dotsc,x_n)$, respectively, and $F$ is
the vector $(f_1,\dotsc,f_m)$. This system has $m$ equations, $k$ free variables (whose values are given) collected in vector $Y$, and $n$ \emph{dependent variables} (or unknowns) collected in vector $X$. Throughout this section , we assume that the $f_i$'s are all of class $\cC^1$ at least.

If the considered system (\ref{eq:F}) is square, i.e., if $m=n$, the \emph{Implicit Function Theorem} (see,
e.g., Theorem $10.2.2$ in \cite{DieudonneEA1}) states that, if
$(\vval{Y},\vval{X})\in\bR^{k+n}$ is a value for the pair $(Y,X)$ such that
$F(\vval{Y},\vval{X})=0$ and the Jacobian matrix of $F$ with respect to $X$  evaluated at $(\vval{Y},\vval{X})$ is nonsingular, then there
exists, in an open neighborhood $U$ of $\vval{Y}$, a unique vector of functions
$G$ such that $F(v,G(v))=0$ for all $v\in{U}$.  In words,
\rref{eq:F} uniquely determines $X$ as a function of $Y$, locally
around $\vval{Y}$. 

Denote by $\Jacobian_{\!X}{F}$ the above-mentioned Jacobian matrix. Solving $F=0$ for $X$, given a value $\vval{Y}$ for $Y$,
requires forming $\Jacobian_{\!X}{F}(\vval{Y})$ as well as inverting it.
One could avoid inverting $\Jacobian_{\!X}{F}$ (or even considering it at all), by focusing instead on its \emph{structural nonsingularity,} introduced in Section~\ref{lkgtuhhihliljfglq}. 

\subsection{{Structural nonsingularity of square matrices}}
\label{lkgtuhhihliljfglq}
Say that a subset of a Euclidian space is \emph{exceptional} if it has 
zero Lebesgue measure.
%empty interior. 
Say that a property $P(x_1,\dots,x_k)$ involving the real variables $x_1,\dots,x_k$ holds \emph{almost everywhere} if it holds for every $x_1,\dots,x_k$ outside an {exceptional} subset of $\bR^k$. 

Square matrix $P$ of size $n$ is a \emph{permutation
  matrix} if and only if there exists a permutation $\sigma$ of the set $\{1,\dots,n\}$ such that $p_{ij}=1$ if $j=\sigma(i)$ and $p_{ij}=0$ otherwise. 
Pre-multiplication (respectively, post-multiplication) of a matrix $A$ by a permutation matrix results in
permuting the rows (resp., the columns) of $A$. 

These preliminary definitions and properties make it possible to prove the following lemma. This result, dealing in particular with the invertibility of sparse matrices, is part of the folklore of High Performance Computing. However, we were unable to find any proper reference in which it is clearly stated (let alone proven).
%\smallskip
%
\begin{lemma} \label {rtgpoupiepo} Let $A$ be an $n{\times}m$-matrix with $m\geq{n}$. The following two properties are equivalent:
\begin{enumerate}
	\item \label{lrukghliruh}
There exist two permutation matrices $P$ and $Q$, such that $P\!AQ=[\,B_1~B_2\,]$, where $B_1$ is square with non-zero entries on its main diagonal---we say that $A$ is \emph{structurally onto};
\item \label{lkughkjyfukgfl}
Matrix $A$ remains almost everywhere onto (surjective) when its non-zero entries vary over some neighborhood. 
%while keeping its pattern of zero-coefficients invariant. 
\end{enumerate}
\end{lemma}
\begin{proof}
%See Appendix~\ref{9og7uiwboq}.
We consider the linear equation $AX=Y$, where $Y\in\bR^n$ has entries $y_1,\dots,y_n$, and $X$ is the unknown vector with real entries $x_1,\dots,x_m$. With a suitable renumbering of the coordinates of $X$, we can assume that $Q$ is the identity matrix. 
For this proof, we will need to formalize what we mean by 
\beq
\mbox{
%\begin{minipage}{8.5cm}
	 ``letting the non-zero entries of matrix $A$ vary over a neighborhood.''
%\end{minipage}
}
\label{ifutyegi}
\eeq
To this end, we consider the \emph{non-zero pattern} of matrix $A$, i.e., the subset $\cP\subseteq[1,n]\times[1,m]$ collecting the pairs $(i,j)$ such that $a_{ij}$ is non-zero; let $K$ be the cardinal of $\cP$. Order the set $\cP$ by lexicographic order, namely $(i,j)<(i',j')$ iff either $i<i'$, or $i=i' \; \land \; j<j'$. Doing this defines an order preserving bijection $\psi:[1,K]\ra{\cP}$. The non-zero entries of matrix $A$ are then collected in the vector $\vvec{A}$ of $\bR^K$ whose components are the $a_{\psi(k)},k=1,\dots,K$ . To every neighborhood $U$ of $\vvec{A}$, we can thus associate the set of $(n{\times}m)$-matrices
\beq
V &=&\left\{
A'=(a'_{ij}) \left|\, a'_{ij} = \xi_{\psi^{-1}(i,j)} \mbox{ if }(i,j)\in{\cP}, \; 0 \mbox{ otherwise} 
\right.\right\}\,,
\label{wujhsytdfujhyt}
\eeq
where the vector of $\bR^K$, whose entries are the $\xi_k$ for ${1{\leq}k{\leq}K}$, ranges over $U$. A set $V$ of matrices obtained in this way is called a \emph{zero-pattern preserving neighborhood of $A$,} or simply \emph{neighborhood of $A$} when ``zero-pattern preserving'' is understood from the context. With this preliminary, we successively prove the two implications.

\myparagraph{\ref{lrukghliruh}$\implies$\ref{lkughkjyfukgfl}}
We need to prove the existence of a zero-pattern preserving neighborhood of $A$ such that, for every $Y$, there exists an $X$ such that $A'X=Y$ holds almost everywhere when matrix $A'$ ranges over this neighborhood. With a suitable renumbering of the equations, we can assume that $P$ is the identity matrix.
 
Let $V_1$ be a zero-pattern preserving neighborhood of $A$ such that the first diagonal term ${a}'_{11}$ of ${A}'=P\!A'$ remains non-zero when $A'$ varies over $V_1$. Since ${a}'_{11}\neq{0}$ holds when $A'$ varies over $V_1$, 
one can perform, on ${A}'$, the first step of a Gaussian elimination by using ${a}'_{11}$ as a pivot. As $x_1$ is thus expressed in terms of $x_2,\dots,x_m,y_1,\dots,y_n$, the first equation may then be removed. 
This  yields a reduced equation $A'_1X_1=Y_1$, where $A'_1$ is  an \mbox{$(n-1){\times}(m-1)$}-matrix, $Y_1\in\bR^{n-1}$ collects $y_{\sigma(2)},\dots,y_{\sigma(n)}$,  and $X_1\in\bR^{m-1}$ collects the unknowns $x_2,\dots,x_m$. Matrix $A'_1$ has diagonal entries equal to  $a^{1\prime}_{kk}={a}'_{kk}-{a}'_{k1}/{a}'_{11}$, for $k=2,\dots,n$. Since ${a}'_{11}$ is non-zero, the set of entries ${a}'_{ij}$ of matrix ${A}'$ causing $a^{1\prime}_{kk}=0$ for some $k$ is exceptional in $\bR^{n{\times}m}$ since it requires the condition ${a}'_{k1}={a}'_{kk}{a}'_{11}$ to hold. The corresponding set of entries $a'_{ij}$ of matrix $A'$ is exceptional as well; we denote by $\Xi_1$ this exceptional subset of $\bR^{n{\times}m}$. Thus, we assume that $A'$ stays within $V_1\setminus\Xi_1$.
	
	This ensures that the diagonal entries of matrix $A'$ are all non-zero. Thus, we can express $x_{2}$ in terms of $x_3,\dots,{x}_{m},y_{\sigma(2)}$. We are then left with a further reduced equation 
  $A'_1X'_1=Y'_1$ for which the same argument applies, namely: if matrix $A'$ stays within a certain zero-pattern preserving neighborhood $V_2$ of $A$, but does not belong to some exceptional set $\Xi_{2}$, then the diagonal entries of $A'_1$ are all non-zero. And so on, by reducing the number of equations and unknowns. 
	
By iterating this process, we prove the existence of sets of matrices $V_1,\dots,V_n$ and $\Xi_{1},\Xi_{2},\dots,\Xi_{n}$ such that, for every matrix belonging to $(V_1\cap\dots\cap{V_n})\setminus(\Xi_{1}\cup\Xi_{2}\cup\dots\cup\Xi_{n})$, the values of $x_1,x_{2},\dots,x_n$ are uniquely determined as functions of the coordinates of $Y$ and the other variables $x_{n+1},\dots,x_m$, by pivoting. This concludes the proof of this implication, since $V_1\cap\dots\cap{V_n}$ is a  zero-pattern preserving neighborhood of $A$ and set $\Xi_{1}\cup\Xi_{2}\cup\dots\cup\Xi_{n}$ is exceptional.
	
\myparagraph{\ref{lkughkjyfukgfl}$\implies$\ref{lrukghliruh}}
%For the only-if part, a
We assume the existence of a zero-pattern preserving neighborhood $V$ of $A$ and an exceptional set $\Xi$ such that $A'$ remains onto when varying over $V\setminus\Xi$. For each onto $A'$, there exists at least one permutation matrix $Q$ such that 
\beq
A'Q&=&\left[\bea{cc}B_1&B_2\eea\right]
\label {wo45tguw}
\eeq
where $B_1$ is a square invertible matrix. Neighborhood $V$ decomposes as the union $V=\bigcup_\sigma V_\sigma$, where $V_\sigma$ collects the matrices $A'$ for which decomposition (\ref{wo45tguw}) holds with the matrix $Q$ representing permutation $\sigma$. Since the number of permutations is finite, at least one of the $V_\sigma$ has non-empty interior, say, for $\sigma=\sigma_*$. We can thus trim $V$ to $V_{\sigma_*}$, which still is a zero-pattern preserving neighborhood of $A$, and we call it again $V$ for convenience.

Hence, there exists a permutation matrix $Q$ such that decomposition (\ref{wo45tguw}) holds for every $A'\in{V}$. By Condition~\emph{\ref{lkughkjyfukgfl}} of Lemma~\ref{rtgpoupiepo}, the  pre-image of $0$ by the map $V\ni{A'}\mapsto{\det}(B_1)$, where $\det(M)$ denotes the determinant of a square matrix $M$, is an exceptional set. 
Now, we have $\det(B_1)=\sum_{i=1}^{n}(-1)^{i+1}{b}_{i1}\det(B_{i1})$, where $B_{ij}$ is the submatrix of $B_1$ obtained by erasing row $i$ and column $j$. Since $B_1$ is almost everywhere invertible, there must be some $i\in\{1,\dots,n\}$ such that $b_{i1}\neq{0}$ and $B_{i1}$ is almost everywhere invertible---otherwise we would have $\det(B_1)=0$ for any ${A'}\in{W}$, where $W$ is some neighborhood contained in $V$. Let $P^1$ be the permutation matrix exchanging rows $1$ and $i$, so we replace $B_1$ by $B_1^1=P^1B_1$. Now, since we have $b^1_{11}=b_{i1}\neq{0}$ and $B^1_{11}$ is almost everywhere invertible. This process is iterated $m-1$ times. For the last $1$-row matrix, its diagonal entry is non-zero, otherwise the determinant of the corresponding $B_1$ matrix would be zero in a neighborhood of $A$. The wanted left permutation matrix is $P=P^{m-1}{\dots}P^2P^1$. 
\end{proof}
Lemma\,\ref{rtgpoupiepo} specializes to the following simpler result:
\begin{corollary} \label {fq09t7ghouiwo} Let $A$ be an $n{\times}n$-matrix. The following two properties are equivalent:
\begin{enumerate}
	\item \label{otuhprouh}
There exist two permutation matrices $P$ and $Q$, such that 
$P\!AQ=B$,
where $B$ has non-zero entries on its main diagonal---we say that $A$ is \emph{structurally nonsingular};
\item \label{krfygrkuygfou}
Matrix $A$ remains almost everywhere invertible when its non-zero entries vary over some neighborhood. 
\end{enumerate}
\end{corollary}
%
%As the proof of Lemma~\ref{rtgpoupiepo} shows, m
%Matrix $A$ is structurally nonsingular if and only if the linear equation $Av=y$, where $v$ is the unknown, has a unique solution for any $y$, for almost all values for the non-zero entries of matrix $A$. 

%\subsection{Background on graphs}
\subsection{Structural analysis of algebraic equations}
\label{sec:erligtuhouih}
In this section, we focus on systems of algebraic equations $F(Y,X)=0$ of the form (\ref{eq:F}), that is, equations involving no time and no dynamics. 
%\begin{notation}\rm 
	%\label{lguitp} In this paper, $\bN=\{0,1,2,\dots\}$ is the set of nonnegative natural integers and $\bR_+=[0,\infty)$ is the set of nonnegative real numbers.\eproof
%\end{notation}
%
We will need to handle variables, their valuations, and vectors thereof. To this end, the following conventions will be used (unless no confusion occurs from using more straightforward notations):
\begin{notation}\rm 
	\label{shltrguioh}  Lowercase letters ($x,y,...$) denote scalar real variables; capitals ($X,Y,...$) denote vectors of real variables; whenever convenient, we regard $X$ as a set of variables and write $x\in{X}$ to refer to an entry of $X$. 
	 We adopt similar conventions for functions ($f,g$,...) and vectors of functions ($F,G$,...).
	 A value for a variable $x$ is generically denoted by $\vval{x}$, and similarly for $X$ and $\vval{X}$.\eproof
\end{notation}
%\subsubsection{Structural nonsingularity of algebraic equations}
%\label{epfuiowerhpfoui}
%Consider a system of smooth algebraic equations:
%\begin{equation}
%f_i(y_1,\dotsc,y_k,x_1,\dotsc,x_n)=0,\quad i=1,\dotsc,m
%\label{eq:F}
%\end{equation}
%rewritten as $F(Y,X)=0$, where $Y$ and $X$ denote the vectors
%$(y_1,\dotsc,y_k)$ and $(x_1,\dotsc,x_n)$, respectively, and $F$ is
%the vector $(f_1,\dotsc,f_m)$. The system has $m$ equations, $k$ input variables collected in vector $Y$, and $n$ \emph{dependent variables} (or unknowns) collected in vector $X$. Throughout this section , we assume that the $f_i$'s are all $\cC^1$ at least.

%If the considered system (\ref{eq:F}) is square, i.e., if $m=n$, the \emph{Implicit Function Theorem} (see, e.g., Theorem $10.2.2$ in \cite{DieudonneEA1}) states that, if $(\vval{Y},\vval{X})\in\bR^{k+n}$ is a value for the pair $(Y,X)$ such that $F(\vval{Y},\vval{X})=0$ and the Jacobian matrix of $F$ with respect to $X$  evaluated at $(\vval{Y},\vval{X})$ is nonsingular, then there exists, in an open neighborhood $U$ of $\vval{Y}$, a unique vector of functions $G$ such that $F(v,G(v))=0$ for all $v\in{U}$.  In words, \rref{eq:F} uniquely determines $X$ as a function of $Y$, locally around $\vval{Y}$. Denote by $\Jacobian_{\!X}{F}$ the above mentioned Jacobian matrix. Solving $F=0$ for $X$, given a value $\vval{Y}$ for $Y$, requires forming $\Jacobian_{\!X}{F}(\vval{Y})$ as well as inverting it.

As explained in the introduction of this section, the existence and uniqueness of solutions for system (\ref{eq:F}) relates to the invertibility of its Jacobian $\Jacobian_{\!X}{F}$. The structural analysis of system (\ref{eq:F}) is built on top of the structural nonsingularity of its Jacobian, introduced in Section~\ref{lkgtuhhihliljfglq}. Its full development requires some background on graphs.

\subsubsection{Background on graphs:}
%We begin with some background on graphs, for which a 
This short background is compiled from~\cite{Berge1962,Hogb06,hogben14}. 
Given a graph $\cG=(V,E)$, where $V$ and $E\subseteq{V}{\times}{V}$ are the sets of vertices and edges, a \emph{matching} $\cM$ of $\cG$ is a set of edges of $\cG$ such that no two edges in $\cM$ are incident on a common vertex. Matching $\cM$ is called \emph{complete} if it covers all the vertices of $\cG$. 
An \emph{$\cM$-alternating path} is a path of $\cG$ whose edges are alternatively in $\cM$ and not in $\cM$---we simply say ``alternating path'' when $\cM$ is understood from the context. 
A vertex is \emph{matched} in $\cM$ if it is an endpoint of an edge in $\cM$, and \emph{unmatched} otherwise. 

A \emph{bipartite graph} is a graph $\cG=(L{\cup}{R},E)$, where $E\subseteq{L}{\times}R$.
Let $F(Y,X)=0$ be a system of equations of the form (\ref{eq:F});  its \emph{bipartite graph} $\cG_F$ is the graph having $F{\cup}{X}$ as its set of vertices and an edge $(f_i,x_j)$ if and only if  variable $x_j$ occurs in function $f_i$. 

\subsubsection{Square systems} We first consider square systems, in which the equations and the dependent variables are in equal numbers.
\begin{definition}[structural nonsingularity]
	\label{erfuilehui} System $F(Y,X)=0$ is called \emph{structurally nonsingular} if its bipartite graph $\cG_F$ possesses a complete matching.
\end{definition}
A structurally nonsingular system is necessarily square, i.e., with equations and dependent variables in equal numbers. Note that Condition~\ref{otuhprouh} of Corollary~\ref{fq09t7ghouiwo} is a matrix reformulation of the graph theoretic condition stated in Definition~\ref{erfuilehui}.
Thus, $F$ is structurally nonsingular in the sense of Definition~\ref{erfuilehui}, 
if and only it its Jacobian $\Jacobian_{\!X}{F}$, evaluated at any pair $(\vval{Y},\vval{X})$ such that $F(\vval{Y},\vval{X})=0$, is structurally nonsingular in the sense of Corollary~\ref{fq09t7ghouiwo}.
The link to numerical regularity is thus formalized in the following lemma:\footnote{This result is implicitly used in~\cite{pantelides,MattssonSoderlin1993,PothenF90,Pryce01}, as well as in a major part of the literature on the structural analysis of DAE systems.}
%, which is a reformulation of Corollary~\ref{fq09t7ghouiwo} using bipartite graphs:\footnote{this result is implicitly used in~\cite{pantelides,MattssonSoderlin1993,PothenF90,Pryce01}.}
%
\begin{lemma}
	\label{oerigtuerio} Assume that $F$ is of class $\cC^1$. The following properties are equivalent:
\begin{enumerate}
	\item System $F(Y,X)=0$ is {structurally nonsingular};
	\item \label{oweruigyweo} For every $(\vval{Y},\vval{X})$ satisfying $F(\vval{Y},\vval{X})=0$, the Jacobian matrix $\Jacobian_{\!X}{F}(\vval{Y},\vval{X})$ remains {generically}\,\footnote{Generically means: outside a set of values of Lebesgue measure zero.\label{erpgfuieh}} nonsingular when its non-zero coefficients vary over some neighborhood.
\end{enumerate}
\end{lemma}  
Property\,{\ref{oweruigyweo}} is interpreted using the argument developed in (\ref{ifutyegi},\ref{wujhsytdfujhyt}): every $n{\times}n$-matrix having the same non-zero pattern as $\Jacobian_{\!X}{F}(\vval{Y},\vval{X})$ identifies with a unique vector of $\bR^K$, where $K$ is the cardinal of this non-zero pattern.
We denote by $\cJ_{\!X}{F}\in\bR^K$ the image of the Jacobian $\Jacobian_{\!X}{F}(\vval{Y},\vval{X})$ obtained via this correspondence. Property\,{\ref{oweruigyweo}} says: 
\beq
\mbox{
\begin{minipage}{11cm}
	 There exist an open neighborhood $U$ of $\cJ_{\!X}{F}$ in $\bR^K$, and a subset $V{\subseteq}{U}$ such that: $(i)$ the set $U\setminus{V}$ has zero Lebesgue measure, and $(ii)$ every $\cJ\in{V}$ yields a regular matrix.
\end{minipage}
} 
\label{leiuthpieu}
\eeq
In the sequel, the so defined sets $U$ and $V$ will be denoted by 
\beq
 \regularU{F}{\vval{Y},\vval{X}} &\mbox{and}& \regularV{F}{\vval{Y},\vval{X}}
 \label{lerigfui}
\eeq
or simply $U_{\!F}$ and $V_{\!F}$ when no confusion can result.\eproof

\smallskip

Let $F(Y,X)=0$ be structurally nonsingular and let $\cM$ be a complete matching for it. Using $\cM$, graph $\cG_F$ can be directed as follows: edge $(f_i,x_j)$ is directed from $f_i$ to $x_j$ if $(f_i,x_j)\in\cM$, from $x_j$ to $f_i$ otherwise. Denote by ${\vec{\cG}}^\cM_F$ the resulting directed graph. The following result holds~(\cite{Duff86}, Chapter 6.10):
\begin{lemma}
	\label{jytdfeoguip} The strongly connected components of ${\vec{\cG}}^\cM_F$ are independent of $\cM$. 
\end{lemma}
Each strongly connected component defines a \emph{block} of $F$, consisting of the set of all equations that are vertices of this component. Blocks are partially ordered and we denote by $\preceq_F$ this order. Extending $\preceq_F$ to a total order (via topological sorting) yields an ordering of equations that puts the Jacobian matrix $\Jacobian_{\!X}F$ in \emph{Block Triangular Form} (BTF).

\begin{figure}[ht]
  \centering
		\centerline{
		\includegraphics[width=17cm]{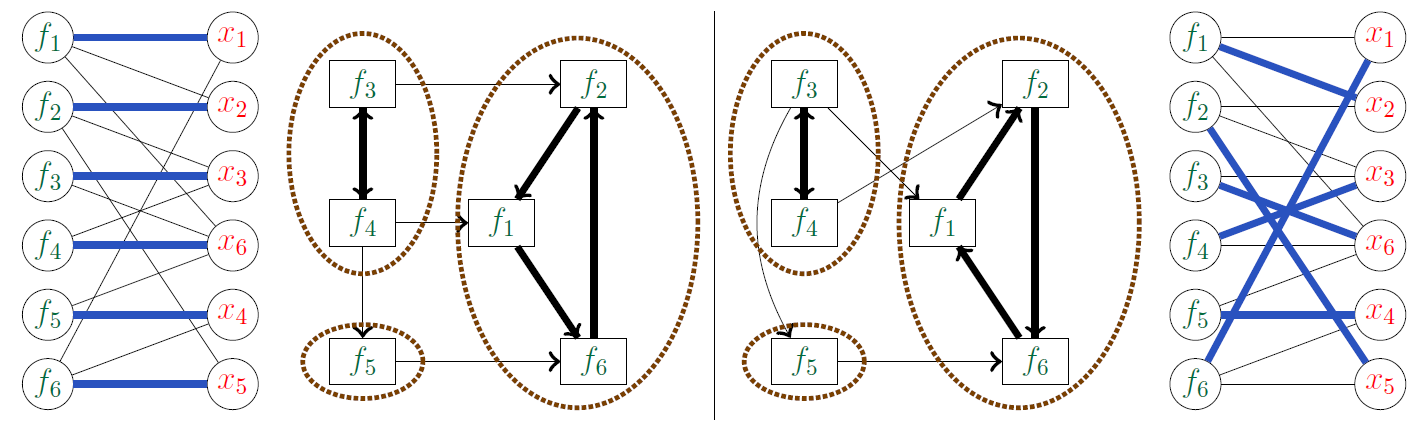}
		}
  \caption{Illustrating Lemma~\ref{jytdfeoguip} (we only show the projection of strongly connected components on function nodes).}
 \label{fig:BTF}
\end{figure}

This lemma is illustrated in \rref{fig:BTF}, inspired by 
\href{https://graal.ens-lyon.fr/~bucar/CR07/lecture-matching.pdf}{lecture notes by J-Y. L'Excellent and Bora U\c car, 2010}.
In this figure, a same bipartite graph $\cG$ is shown twice (left and right), with the equations sitting on the left-hand side in green, and the variables on the right-hand side in red. Two complete matchings $\cM_1$ (left) and $\cM_2$ (right) are shown in thick blue, with other edges of the bipartite graph being black. The restriction, to the equation vertices, of the two directed graphs ${\vec{\cG}}^{\cM_1}_F$ (left) and ${\vec{\cG}}^{\cM_2}_F$ (right) are shown on both sides.
Although the two directed graphs ${\vec{\cG}}^{\cM_1}_F$ and ${\vec{\cG}}^{\cM_2}_F$ differ, the resulting block structures (encircled in dashed brown) are identical.

\subsubsection{Non-square systems, Dulmage-Mendelsohn decomposition}
\label{fltorghbfdio}
In our development, we will also encounter non-square systems of algebraic equations. For general systems (with a number of variables not necessarily equal to the
 number of equations, $n{\neq}{m}$ in (\ref{eq:F})), call \emph{block} a pair $\block=(\bff,\bfx)$, consisting of a subsystem $\bff{=}0$ of system $F{=}0$ and the subset $\bfx$ of its dependent variables. 
\begin{definition}[Dulmage-Mendelsohn]
	\label{elrftuierhfperui}  For $F=0$ a general system of algebraic equations, the \emph{Dulmage-Mendelsohn decomposition~\cite{PothenF90}} of $\cG_F$ yields the partition of system $F=0$ into the following three blocks, given some matching of maximal cardinality for $\cG_F$:
\begin{itemize}

\vspace*{-2mm}

 \item 
block $\block_\overapprox$ collects the variables and equations reachable via some alternating path from some unmatched equation;

\vspace*{-2mm}

	\item 
  block $\block_\underapprox$ collects the variables and equations reachable via some alternating path from some unmatched variable;
  
\vspace*{-2mm}

 \item 
block $\block_\squared$ collects the remaining variables and equations. 
  
\vspace*{-1mm}

\end{itemize}
Blocks $\block_\overapprox,\block_\underapprox,\block_\squared$ are the $\overapprox$verdetermined, $\underapprox$nderdetermined, and $\squared$nabled parts of $F$.
\end{definition}
Statement~\ref{erpoigfehpguio} of the following lemma ensures that Definition~\ref{elrftuierhfperui} is meaningful:
\begin{lemma}
	\label{erlihpfiu} \
\begin{enumerate}
	\item \label{erpoigfehpguio} The triple $(\block_\underapprox,\block_\squared,\block_\overapprox)$ defined by the Dulmage-Mendelsohn decomposition does not depend on the particular maximal matching used for its definition. 
We thus write 
\beq
\DM(F)&=&(\block_\underapprox,\block_\squared,\block_\overapprox)\,.
\label{etrgoijrtpkujyoi}
\eeq
\item System $F$ is structurally nonsingular if and only if the  overdetermined and underdetermined blocks $\block_\overapprox$ and $\block_\underapprox$ are both empty.
\end{enumerate}
\end{lemma}
A refined DM decomposition exists, which in addition refines block $\block_\squared$ into its indecomposable block triangular form. This is a direct consequence of applying Lemma~\ref{jytdfeoguip} to $\block_\squared$ in order to get its BTF.
The following lemma is useful in the structural analysis of multimode DAE systems~\cite{BENVENISTE2020,DBLP:journals/corr/abs-2008-05166}:
\begin{lemma}
\label{togiethgouiho} Let $\DM({F})=(\block_\underapprox,\block_\squared,\block_\overapprox)$ be the Dulmage-Mendelsohn decomposition of $F=0$. Then, no overdetermined block exists in the Dulmage-Mendelsohn decomposition of ${F}\setminus\block_\overapprox$.
\end{lemma}
\begin{ccomment}\rm 
	\label{esrltuiophu} 
	We could refine this result by only removing unmatched equations with reference to the matching of maximal cardinality used for generating $\DM({F})$. Then, denoting by $F'$ the remaining subsystem of $F$, we would still get that no overdetermined block exists in the Dulmage-Mendelsohn decomposition of $F'$, and $F'\supseteq({F}\setminus\block_\overapprox)$ (the inclusion being strict in general). However, this policy depends on the particular matching. In contrast, the rule defined by Lemma~\ref{togiethgouiho} is independent from the matching. This Lemma is used in~\cite{BENVENISTE2020,DBLP:journals/corr/abs-2008-05166} to generate restart code at mode changes, see Section~\ref{ujydtfweutyf} of this report; the independence with respect to the particular matching considered was crucial to ensure the determinism of code generation.
%
%Since we will rely on Lemma~\ref{togiethgouiho} for our compilation scheme in \rref{sec:loguhiuliugh}, adopting the finer policy would introduce a risk of nondeterminism due to the arbitrary choice of the particular complete matching selected. This is the reason for sticking with the policy stated in Lemma~\ref{togiethgouiho}.
%\eproof
\end{ccomment}

 \subsection{Practical use of structural analysis}
We distinguish between two different uses of structural analysis. The classical use assumes that we already have a close guess of a solution, we call it the \emph{local use.} By contrast, the other use is called \emph{global.} We explain both here.

\subsubsection{Local use of structural analysis} 
If a square matrix is structurally singular, then it is singular. The converse is false: structural nonsingularity does not guarantee nonsingularity. Therefore, the practical use of Definition\,\ref{erfuilehui} is as follows: We first check if $F=0$ is structurally nonsingular. If not, then we abort searching for a solution. Otherwise, we then proceed to computing a solution, which may or may not succeed depending on the actual numerical regularity of the Jacobian matrix $\Jacobian_{\!X}F$.

Structural analysis relies on the Implicit Function Theorem for its justification, as stated at the beginning of this section. Therefore, the classical use of structural analysis is local: let $F(Y,X)=0$ be a structurally nonsingular system of equations with dependent variables $X$, and let $\vval{X}$ satisfy $F(\vval{Y},\vval{X})=0$ for a given value $\vval{Y}$ for $Y$; then, the system remains (generically) nonsingular if $\vval{}$ varies in a neighborhood of $\vval{Y}$. This is the situation encountered in the running of ODE or DAE solvers, since a close guess of the solution is known from the previous time step.

\subsubsection{Non-local use of structural analysis}
\label{wlerifuhwlpi}
We are, however, also interested in invoking the structural nonsingularity of system $F(Y,X)=0$ when no close guess is known. This is the situation encountered when handling mode changes of multimode DAE systems~\cite{BENVENISTE2020,DBLP:journals/corr/abs-2008-05166}. We thus need another argument to justify the use of structural analysis in this case.

As a prerequisite, we recall basic definitions regarding smooth manifolds (see, e.g.,~\cite{Cartan}, Section 4.7). In the next lemma, we consider the system $F(Y,X)=0$ for a fixed value of $Y$, thus $Y$ is omitted. 
%\clearpage
\begin{lemma}
	\label{rleiufheliru} Let $k,m,n\in\bN$ be such that $m<{n}$ and set $p=n-m$. For $\cS\subset\bR^n$ and $x^*=(x^*_1,\dots,x^*_n)\in\cS$, the following properties are equivalent:
\begin{enumerate}
	\item \label{lreiuhleuirh} There exists an open neighborhood $V$ of $x^*$ and a $\cC^k$-diffeomorphism $F:V\ra{W}{\subset}\bR^n$ such that $F(x^*)=0$ and $F(\cS\cap{V})$ is the intersection of $W$ with the subspace of $\bR^n$ defined by the $m$ equations $w_{p+1}=0,\dots,w_n=0$, where $w_i$ denote the coordinates of points belonging to $W$.
	
	\item  \label{weuygweouertpoi} There exists an open neighborhood $V$ of $x^*$, an open neighborhood $U$ of $0$ in $\bR^p$ and a homeomorphism $\psi:U\rightarrow\cS\cap{V}$, such that $\psi$, seen as a map with values in $\bR^n$, is of class $\cC^k$ and has rank $p$ at $0$---i.e., $\psi'(0)$ has rank $p$.
\end{enumerate}
\end{lemma}
Statement\,\ref{lreiuhleuirh} means that $\cS$ is, in a neighborhood of $x^*$, the solution set of the system of equations $f_1(X)=0,\dots,f_m(X)=0$ in the $n$-tuple of dependent variables $X$, where $F=(f_1,\dots,f_m)$. Statement\,\ref{weuygweouertpoi} expresses that $\psi$ is a parameterization of $\cS$ in a neighborhood of $x^*$, of class $\cC^k$ and rank $p$---meaning that $\cS$ has dimension $p$ in a neighborhood of $x^*$.

{In order to apply Lemma\,$\ref{rleiufheliru}$ to the non-local use of structural analysis, we have to study the case of square systems of equations (i.e., let $m = n$).}
Consider a system $F(X)=0$, where $X=(x_1,\dots,x_n)$ is the $n$-tuple of dependent variables and $F=(f_1,\dots,f_n)$ is an $n$-tuple of functions $\bR^n\rightarrow\bR$ of class $\cC^k$. Let $\cS\subseteq\bR^n$ be the solution set of this system, {that we assume is non-empty}. To be able to apply Lemma\,\ref{rleiufheliru}, we augment $X$ with one extra variable $z$, i.e., we set $Z\eqdef{X}\cup\{z\}$ and we now regard our formerly square system $F(X)=0$ as an augmented system $G(Z)=0$ where $G(X,z)\eqdef{F}(X)$. Extended system $G$ possesses $n$ equations and $n{+}1$ dependent variables, hence, $p=1$,  and its solution set is equal to $\cS\times\bR$. 
Let $x^* \in \cS$ be a solution of $F=0$. Then, we have $G(x^*,z^*)=0$ for any $z^*\in\bR$. Assume further that the Jacobian matrix $\Jacobian_{\!X}F(x)$ is nonsingular at $x^*$. Then, it remains nonsingular in a neighborhood of $x^*$. Hence, the Jacobian matrix $\Jacobian_{\!X}G(x,z)$ has rank $n$ in a neighborhood $V$ of $(x^*,z^*)$ in $\bR^{n+1}$. We can thus extend $G$ by adding one more function in such a way that the resulting $(n{+}1)$-tuple $\bar{G}$ is a $\cC^k$-diffeomorphism, from $V$ to an open set $W$ of $\bR^{n+1}$. The extended system $\bar{G}=0$ satisfies Property\,\ref{lreiuhleuirh} of Lemma\,\ref{rleiufheliru}. By Property\,\ref{weuygweouertpoi} of Lemma\,\ref{rleiufheliru}, $(\cS\times\bR)\cap{V}$ has dimension $1$, implying that $\cS\cap\proj{\bR^n}{V}$ has dimension $<1$. This analysis is summarized in the following lemma:
\begin{lemma}
	\label{erkltuiherliu}  Consider the square system $F(X)=0$, where $X=(x_1,\dots,x_n)$ is the $n$-tuple of dependent variables and $F=(f_1,\dots,f_n)$ is an $n$-tuple of functions $\bR^n\rightarrow\bR$ of class $\cC^k, k{\geq}1$. Let $\cS\subseteq\bR^n$ be the solution set of this system. Assume that system $F=0$ has solutions and let $x^*$ be such a solution. Assume further that the Jacobian matrix $\Jacobian_{\!X}F(x)$ is nonsingular at $x^*$. Then, there exists an open neighborhood $V$ of $x^*$ in $\bR^n$, such that $\cS\cap{V}$ has dimension $<1$.
\end{lemma}
{The implications for the non-local use of structural analysis are as follows.} Consider the square system $F(X)=0$, where $X=(x_1,\dots,x_n)$ is the $n$-tuple of dependent variables and $F=(f_1,\dots,f_n)$ is an $n$-tuple of functions $\bR^n\rightarrow\bR$ of class $\cC^k$. 
We first check the structural nonsingularity of $F=0$. If structural nonsingularity does not hold, then we abort solving the system. Otherwise, we proceed to solving the system and the following cases can occur:
\begin{enumerate}
	\item \label{epriuthpeo} the system possesses no solution;
	
	\item \label{kerufygeouy} its solution set is nonempty, of dimension $<1$; or
	
	\item \label{leriughpoiuh} its solution set is nonempty, of dimension $\geq 1$.
\end{enumerate}
Based on Lemma\,\ref{erkltuiherliu}, one among cases\,\ref{epriuthpeo} or\,\ref{kerufygeouy} will hold, generically, whereas case\,\ref{leriughpoiuh} will hold exceptionally. Thus, structurally, we know that the system is not underdetermined. 
This is our justification of the non-local use of structural analysis. Note that the existence of a solution is not guaranteed. Furthermore, unlike for the local use, uniqueness is not guaranteed either. {{Lemma\,\ref{erkltuiherliu} only states that there cannot be two arbitrarily close solutions of $F=0$, as its conclusion $\cS \cap V = \{ x^* \}$ involves an open neighborhood $V$ of $x^*$.}} Of course, subclasses of systems for which existence and uniqueness are guaranteed are of interest.

\subsection{Equations with existential quantifiers
\label{sec:operiuhlo}}
To support the structural analysis for the method of \emph{differential arrays} advocated by Campbell and Gear~\cite{CampbellGear1995} (see also Section~\ref{lrtioghprtjpi} of this report), we will need to develop a structural analysis for the following class of (possibly non-square) algebraic systems of equations with existential quantifier:\footnote{To our knowledge, the results of this section are not part of the folklore of sparse matrix algebra.}
	  \beq
	 \exists{W}:F(X,W,Y)=0 ~ ,
	 \label{elrifuhuio}
	 \eeq
	  where $(X,W)$ collects the dependent variables of system $F(X,W,Y)=0$. $Y$ and $X$ are as before, and the additional tuple $W$ collects supplementary variables, of no interest to us, whence their elimination by existential quantification.
Our aim is to find structural conditions ensuring that (\ref{elrifuhuio}) defines a partial function $Y\ra{X}$, meaning that the value of tuple $X$ is uniquely defined by the satisfaction of (\ref{elrifuhuio}), given a value for tuple $Y$.
	   
	   At this point, a fundamental difficulty arises. Whereas (\ref{elrifuhuio}) is well defined as an abstract relation, it cannot in general be represented by a projected system of smooth algebraic equations of the form $G(X,Y)=0$. Such a $G$ can be associated to (\ref{elrifuhuio}) only in subclasses of systems.\footnote{Examples are linear systems for which elimination is easy, and polynomial systems for which it is doable---but expensive---using Gr\"obner bases.} In general, no extension of the Implicit Function Theorem exists for systems of the form (\ref{elrifuhuio}); thus, one cannot apply as such the arguments developed in Section~\ref{wlerifuhwlpi}. Therefore, we will reformulate our requirement differently.
	   Say that 
		\beq
		\mbox{
		\begin{minipage}{10cm}
			 $\vval{Y}$ is \emph{consistent} if the system of equations $F(X,W,\vval{Y})=0$ possesses a solution for $(X,W)$.
		\end{minipage}
		} \label{elrgihsuilpn}
		\eeq
		Consider the following property for the systen $F(X,W,Y)=0$: for every consistent $\vval{Y}$,
	   \beq
	  \left.\bea{l}
	  F(\vval{X}^1,\vval{W}^1,\vval{Y})=0 \\ [1mm]
	  F(\vval{X}^2,\vval{W}^2,\vval{Y})=0
	  \eea\right\}\implies& \vval{X}^1=\vval{X}^2 ~~ , 
	   \label{lriuyokiyugti}
	   \eeq
	   expressing that $X$ is independent of $W$ given a consistent tuple of values for $Y$. 
	   To find structural criteria guaranteeing (\ref{lriuyokiyugti}), we will consider the algebraic system $F(X,W,Y)=0$ with $X,W,Y$ as dependent variables (i.e., values for the entries of $Y$ are no longer seen as being given).
\begin{definition}
	\label{lerigfueriopg} System $(\ref{elrifuhuio})$ is \emph{structurally nonsingular} if the following two conditions hold, almost everywhere when the non-zero coefficients of the Jacobian matrix $\Jacobian_{X,W,Y}F$ vary over some neighborhood:\footnote{The precise meaning of this statement is the same as in Lemma~\ref{oerigtuerio}, see the discussion therafter.} 
	\begin{itemize}
	\item[$\bullet$] consistent values for $Y$ exist;
	\item[$\bullet$] condition $(\ref{lriuyokiyugti})$ holds.
	\end{itemize}
\end{definition}
The structural nonsingularity of (\ref{elrifuhuio}) can be checked by using the DM decomposition of $F$ as follows. Let $(\block_\underapprox,\block_\squared,\block_\overapprox)=\DM(F)$ be the DM decomposition of $F(X,W,Y)=0$, with $(X,W,Y)$ as dependent variables. We further assume that the regular block $\block_\squared$ is expressed in its block triangular form, see Lemma~\ref{jytdfeoguip} and comments thereafter. Let $\bfB$ be the set of all indecomposable blocks of $\block_\squared$ and $\preceq$ be the partial order on $\bfB$ following Lemma~\ref{jytdfeoguip}. The following holds:
\begin{lemma}
	\label{lergfuioerhpuip} System $(\ref{elrifuhuio})$ is \emph{structurally nonsingular} if and only if:
\begin{enumerate}
	\item \label{sguiholghouiy} Block $\block_\overapprox$ is empty;
	\item \label{odigpiufr} Block $\block_\underapprox$ involves no variable belonging to $X$;
	\item \label{lpiguhlosgioj} For every $\block\in\bfB$ containing some variable from $X$, then, $\block$ contains no variable from $W$, and for every $\block'{\prec}\block$ and every directed edge $(z',{z}){\in}\vec{\cG}^\cM_F$ such that $z'{\in}\block'$, $z{\in}\block$, and $\cM$ is an arbitrary complete matching for $\cG_F$, then 	$z'\not\in{W}$.
\end{enumerate}
\end{lemma}
Conditions~\ref{sguiholghouiy} and~\ref{odigpiufr} speak by themselves. Regarding Condition~\ref{lpiguhlosgioj}, the intuition is the following. Every directed path of $\vec{\cG}^\cM_F$, originating from $W$ and terminating in $X$, of minimal length, must traverse $Y$. Consequently, $W$ influences $X$ ``through'' $Y$ only.
\begin{proof} We successively prove the if and then the only if parts.

\myparagraph{``If'' part} By Condition~\ref{sguiholghouiy}, there exist consistent values for $Y$. By condition~\ref{odigpiufr}, no variable of $X$ belongs to the underdetermined part of $\DM(F)$. It remains to show that condition $(\ref{lriuyokiyugti})$ holds. By condition (\ref{lpiguhlosgioj}), each indecomposable block $\block\in\bfB$ involving a variable of $X$ has the form $G(Z,U)=0$,
where 
\begin{enumerate}
	\item $Z$ is the $n$-tuple of dependent variables of $G=(g_1,\dots,g_n)$, 
	\item no variable of $W$ belongs to $Z$, and
	\item $U\subseteq{X}{\cup}Y$ is a subset of the dependent variables of the blocks $\block'$ immediately preceding $\block$.
\end{enumerate}
Since $G=0$ is structurally regular, fixing a value for $U$ entirely determines all the dependent variables of block $\block$. This proves the ``if'' part.

\myparagraph{``Only if'' part} We prove it by contradiction. If condition~\ref{sguiholghouiy} does not hold, then the existence of consistent values for $Y$ is not structurally guaranteed. If condition~\ref{odigpiufr} does not hold, then the variables of $X$ that belong to $\block_\underapprox$ are not determined, even for given values of $W,Y$. If condition~\ref{lpiguhlosgioj} does not hold, then two cases can occur:
\begin{itemize}
	\item Some $x\in{X}$ and $w\in{W}$ are involved in a same indecomposable block $\block\in\bfB$. Then, condition (\ref{lriuyokiyugti}) will not hold for $x$ due to the structural coupling with $w$ through block $\block$.
	\item There exists an indecomposable block $\block\in\bfB$  of the form $G(Z,U)=0$,
where 
\begin{enumerate} 
	\item $Z$ is the $n$-tuple of dependent variables of $G=(g_1,\dots,g_n)$, and $Z$ contains a variable $x\in{X}$,
	\item no variable of $W$ belongs to $Z$, and
	\item some variable $w\in{W}$ belongs to $U$, the set of the dependent variables of the blocks $\block'$ immediately preceding $\block$.
\end{enumerate}
\end{itemize}
Hence, $w$ influences $x$, structurally, which again violates (\ref{lriuyokiyugti}). This finishes the proof.
\end{proof}
\begin{algorithm}[ht]
	\caption{$\atomicact{ExistQuantifEqn}$} 
	\label{alg:erofiueopiu}
 \begin{algorithmic}[1]
\Require $F$;
\Return $(b_\overapprox,b_\underapprox,F_\Sigma,\consistency{F}_\Sigma)$
\State $(\block_\underapprox,\block_\squared,\block_\overapprox)=\DM(F)$
\If{condition~\ref{sguiholghouiy} of Lemma~\ref{lergfuioerhpuip} holds} 
\State $b_\overapprox\gets\ttt$;
\If{\ref{odigpiufr} and \ref{lpiguhlosgioj} of Lemma~\ref{lergfuioerhpuip} hold} 
\State $b_\underapprox\gets\ttt$;
  \State partition $\block_\squared=F_\Sigma\cup\consistency{F}_\Sigma$
\label{op:efuiehoi}
\Else
\State $b_\underapprox\gets\fff$
\EndIf
\Else
\State $b_\overapprox\gets\fff$
\EndIf
	\end{algorithmic}
\end{algorithm}
We complement Lemma~\ref{lergfuioerhpuip} with the algorithm $\atomicact{ExistQuantifEqn}$ (\rref{alg:erofiueopiu}), which
requires a system $F$ of the form (\ref{elrifuhuio}). If condition~\ref{sguiholghouiy} of Lemma~\ref{lergfuioerhpuip} fails to be satisfied, then $b_\overapprox\gets\fff$ is returned, indicating overdetermination. If conditions~\ref{odigpiufr} or \ref{lpiguhlosgioj} of Lemma~\ref{lergfuioerhpuip} fail to be satisfied, then $b_\underapprox\gets\fff$ is returned, indicating underdetermination. Otherwise, $\atomicact{ExistQuantifEqn}$ succeeds and returns the value $\ttt$ for both Booleans, together with
the decomposition $F_\Sigma\cup\consistency{F}_\Sigma$ of $\block_\squared$. In this decomposition:
\begin{itemize}
\item Subsystem $F_\Sigma$ collects the indecomposable blocks involving variables belonging to $X$, so that ${F}_\Sigma$ determines $X$ as a function of $\vval{Y}$ when $\vval{Y}$ is consistent;
\item Subsystem $\consistency{F}_\Sigma$ collects the consistency conditions, whose dependent variables belong to $W\cup{Y}$.
\end{itemize}
Our background on the structural analysis of algebraic equations is now complete. 
In the next section, we recall the background on the structural analysis of (single-mode) DAE systems.

\section{Structural analysis of DAE systems}
\label{lfdkuvekrughluh}
In this section, we consider DAE systems of the form 
\beq
f_j(x_i \mbox{'s and derivatives}) = 0\,,
\label{sepguiohp}
\eeq
where $x_1,\dotsc,x_m$ denote the dependent signal variables (their valuations are trajectories) and $f_1=0,\dotsc,f_n=0$ denote the equations---with reference to notation (\ref{eq:F}), we omit the input signals collected in $Y$.
By analogy with DAEs, we use the acronym dAE to mean \emph{difference Algebraic Equations,} which define discrete-time dynamical systems of the form 
\beq
f_j(x_i \mbox{'s and  shifts}) = 0\,.
\label{ltrgiourhtuil}
\eeq
where $x_1,\dotsc,x_m$ denote the dependent stream variables (their valuations are data streams, which can be regarded as real sequences), $f_1=0,\dotsc,f_n=0$ denote the equations, and,
for $x=\{x_n\mid{n}=1,2,\dots\}$ a stream of variables, its successive \emph{shifts} $\ppostset{k}{x}$, where $k$ is a nonnegative integer, are defined by 
\beq
\ppostset{k}{x}_n \eqdef x_{n+k} \enspace ; 
&\mbox{we write for short}&  \postset{x} \eqdef \ppostset{1}{x}\,.
\label{lrigkuhkjjhvahtr}
\eeq
Also, the notation $\pprime{k}{x}$ is adopted throughout this paper, instead of the more classical $x^{(k)}$, for the $k$-th derivative of $x$.
%
%\begin{notation} \rm
	%\label{lerigftuheroiuof}
%Let $\Vars$ be an underlying set of variables. For
%$x{\in}\Vars$ and $m{\in}\bN$, the $m$-differentiation and $m$-shift of $x$ are denoted by
%$\pprime{m}{x}$ and $\ppostset{m}{x}$, respectively.  
%Let 
%$\pprime{m}{\Vars}$ and $\ppostset{m}{\Vars}$ denote the set of all
%$\pprime{m}{x}$ and $\ppostset{m}{x}$, for $x$ ranging over the set
%$\Vars$ of variables. Let
%\beq
%\Vars^{(\prime)}\eqdef\bigcup_{m\in\bN}\pprime{m}{\Vars} \mbox{ and }
%\Vars^{(\bullet)}\eqdef\bigcup_{m\in\bN}\ppostset{m}{\Vars} \enspace .
%\label{erpfuhqf9}
%\eeq
%For $X{\subset}\Vars$, sets $\pprime{m}{X},\ppostset{m}{X},X^{(\prime)},X^{(\bullet)}$ are defined in a similar way.\eproof
%\end{notation}

\subsection{John Pryce's \sigmamethod\ for square systems}
\label{sec:structural}
\label{reoiughfuil}
The DAE systems we consider are ``square'', i.e., have the form (\ref{sepguiohp}), with $m=n$. Call \emph{leading variables} of System~(\ref{sepguiohp}) the $d_i$-th derivatives $\pprime{d_i}{x_i}$ for $i=1,\dots,n$, where $d_i$ is the maximal differentiation degree of variable $x_i$ throughout $f_1=0,\dotsc,f_n=0$. The problem addressed by the structural analysis of DAE systems of the form (\ref{sepguiohp}) is the following. Regard (\ref{sepguiohp}) as a system of algebraic equations with the leading variables as unknowns. If this system is structurally nonsingular, then, given a value for all the $\pprime{k}{x_i}$ for $i=1,\dots,n$ and $k=0,\dots,d_i{-}1$, a unique value for the leading variables can be computed, structurally; hence, System~(\ref{sepguiohp}) is ``like an ODE''. If this is not the case, finding additional \emph{latent equations} by differentiating suitably selected equations from (\ref{sepguiohp}) will bring the system to an ODE-like form, while not changing its set of solutions. Performing this is known as \emph{index reduction.} 

%In our simple examples, finding latent equations was easy. 
Algorithms were proposed in the literature for doing it efficiently. Among them, the Pantelides' algorithm~\cite{pantelides} is the historical solution. We decided, however, to base our subsequent developments on the beautiful method proposed in~\cite{Pryce01} by J. Pryce, called the \sigmamethod. The \emph{\sigmamethod} also covers the construction of the block triangular form and addresses numerical issues, which we do not discuss here. 

\paragraph{Weighted bipartite graphs:}
We consider System~(\ref{sepguiohp}), which is entirely characterized by its set of dependent variables $X$ (whose generic element is denoted by $x$) and its set of equations $F=0$ (whose generic element is written $f=0$). 
We attach to (\ref{sepguiohp}) the bipartite graph $\cG=(F\cup{X},\Edges{\cG})$ having an edge $(f,x)\in{\Edges{\cG}}$ if and only if $x$ occurs in function $f$, regardless of its differentiation degree. Recall that a matching is \emph{complete} iff it involves all equations of $F$.

So far, $\cG$ is agnostic with respect to differentiations. To account for this, we further equip $\cG$ with \emph{weights}: to each edge $(f,x)\in\Edges{\cG}$ is associated a nonnegative integer $d_{\!f\!x}$, equal to the maximal differentiation degree of variable $x$ in function $f$. This yields a \emph{weight} for any matching $\cM$ of $\cG$ by the formula $w(\cM)=\sum_{(f,x)\in\cM}d_{\!f\!x}$.
Suppose we have a solution to the following problem:
\begin{problem}
	\label{liftuerhpituhepu8} 
Find a complete matching $\cM$ for $\cG$ and nonnegative integer \emph{offsets $\{c_f\mid{f\in{F}}\}$ and $\{d_x\mid{x\in{X}}\}$}, satisfying the following conditions, where the $d_{\!f\!x}$ are the weights as before:
\beq
\bea{rcll}
d_x-c_f &\!\!\!\geq\!\!\!& d_{\!f\!x} & \mbox{for all $(f,x)\in\Edges{\cG}$, with equality if } (f,x){\in}\cM
 \\
c_f &\!\!\!\geq\!\!\!& 0  & \mbox{for all } f\in{F}\,.
\eea
\label{ltuhltrli}
\eeq
\end{problem}
Then, differentiating $c_f$ times each function $f$ yields a DAE system $F_\Sigma=0$ having the following properties. Inequality $d_x\geq{c}_f+d_{\!f\!x}$ holds for each $x{\in}X$, and $d_x={c}_f+d_{\!f\!x}$ holds for the unique $f$ such that  $(f,x){\in}\cM$. Hence, the leading variables of DAE system $F_\Sigma=0$ are the 
$d_x$-th derivatives $\pprime{d_x}{x}$. Consequently, system $F_\Sigma=0$, now seen as a system of algebraic equations having $\pprime{d_x}{x}$ as dependent variables, is structurally nonsingular by Definition~\ref{erfuilehui}. Hence, $F_\Sigma=0$ is ``like an ODE''. The integer $k=\max_{f\in{F}}\,c_f$ is called the \emph{index} of the system.\footnote{We should rather say the \emph{differentiation index} as, once again, other notions of index exist for DAEs~\cite{CampbellGear1995} that are not relevant to our work. Also note that the standard definition of the differentiation index, used in Pryce's article~\cite{Pryce01}, slightly differs from the one that we adopted in this report for the sake of clarity.}
\begin{definition}
	\label{def:erlgtuioo} For $F$ a DAE system, the solution to Problem\,$\ref{liftuerhpituhepu8}$ yields the DAE system $F_\Sigma=$ $\{f^{\prime{c_f}}\mid f{\in}{F}\}$ together with the system of \emph{consistency constraints} $\overline{F}_\Sigma=\{f^{\prime{k}}\mid f{\in}{F},0{\leq}{k}{<}c_f\}$.
\end{definition}
Knowing the offsets also allows transforming $F$ into a system of index~$1$, by not performing the final round of differentiations. 
There are infinitely many solutions to Problem~\ref{liftuerhpituhepu8} with unknowns $c_f$ and $d_x$, since, for example, adding the same $\ell\in\bN_{\geq{0}}$ to all $c_f$'s and $d_x$'s yields another solution. We thus seek for a \emph{smallest} solution, elementwise.
Hence, Problem~\ref{liftuerhpituhepu8} is the fundamental problem we must solve, and we seek a smallest solution for it.

The beautiful idea of J. Pryce is to propose a linear program encoding Problem~\ref{liftuerhpituhepu8}. As a preliminary step, we claim that a bruteforce encoding of Problem~\ref{liftuerhpituhepu8} is the following: Find nonnegative integers $\xi_{f\!x},d_x,c_f$ such that
\begin{equation}
\bea{rl}
\left.\bea{r}
\sum_{f:(f,x)\in\Edges{\cG}}\;\xi_{f\!x}=1 \\ 
\sum_{x:(f,x)\in\Edges{\cG}}\;\xi_{f\!x}={1} \\ 
\xi_{f\!x}\geq{0} \eea\right\}\hspace*{-3mm} &\mbox{complete matching}
\\ [6mm]
\left.\bea{r}
d_x{-}c_f-d_{\!f\!x}\geq{0} \\ 
c_f\geq{0} \eea\right\}\hspace*{-3mm}  &\mbox{encodes ``$\geq$'' in (\ref{ltuhltrli})}
\\ [3mm]
\displaystyle\sum_{(f,x)\in\Edges{\cG}}\;\xi_{f\!x}(d_x{-}c_f-d_{\!f\!x})=0\hspace*{-1mm}  &\mbox{encodes ``$=$'' in (\ref{ltuhltrli})}
\eea
\label{kerfyugwerluh}
\end{equation}
where, for the formulas having no summation: $x,f$, and $(f,x)$ range over $X,F$, and $\Edges{\cG}$ respectively.
We now justify our claim. Focus on the first block. Since the $\xi$'s are nonnegative integers, they can only take values in $\{0,1\}$ and one defines a subgraph of graph $\cG$ by only keeping edges $(f,x)$ such that $\xi_{f\!x}=1$. The first two equations formalize that the chosen subset of edges is a matching, which  is complete since all vertices of $\cG$ are involved. The second block is a direct encoding of (\ref{ltuhltrli}) if we ignore the additional statement ``with equality iff''. The latter is encoded by the last constraint (since having the sum equal to zero requires that all the terms be equal to zero).
Constraint problem (\ref{kerfyugwerluh}) does not account for our wish for a ``smallest'' solution: this will be handled separately.

Following the characterization of solutions of linear programs via  \emph{complementary slackness conditions}, every solution of problem (\ref{kerfyugwerluh}) is a solution of the following dual linear programs (LP), where the $\xi_{f\!x}$ and the $c_f$ are real:
\beq\mbox{\emph{primal}} \hspace*{-3mm}&:&\bea{rl}
\mbox{maximize}&\sum_{(f,x)\in\Edges{\cG}}\;d_{\!f\!x}\;\xi_{f\!x} \\
[1mm] \mbox{subject to}&\sum_{f:(f,x)\in\Edges{\cG}}\;\xi_{f\!x}=1 \\
\mbox{and}&\sum_{x:(f,x)\in\Edges{\cG}}\;\xi_{f\!x}\geq{1} \\
\mbox{and}&\xi_{f\!x}\geq{0}
\eea
\label{erofuierhopiu}
\\ [2mm]
\mbox{\emph{dual}} \hspace*{-3mm}&:&\bea{rl}
\mbox{minimize}&\sum_xd_x{-}\sum_f\;c_f \\
[1mm] \mbox{subject to}&d_x{-}c_f\geq{d}_{\!f\!x} \\
\mbox{and}&c_f\geq{0}
\eea
\label{guiohpui}
\eeq
where, for the formulas having no summation: $x,f$, and $(f,x)$ range over $X,F$, and $\Edges{\cG}$ respectively.
In these two problems, $f$ ranges over $F$, $x$ ranges over $X$, and $(f,x)$ ranges over $\cG$. Also, we have relaxed the integer LP to a real LP, as all solutions to the integer LP are solutions of the real LP. Note that LP (\ref{erofuierhopiu}) encodes the search for a complete matching of maximum weight for $\cG$. By the principle of complementary slackness in linear programming, 
\beq
\mbox{
\begin{minipage}{7cm}
	 for respective optima of problems (\ref{erofuierhopiu}) and (\ref{guiohpui}), $\xi_{f\!x}{>}0$ if and only if $d_x{-}c_f=d_{\!f\!x}$,
\end{minipage}
}
\label{eogiheoiwy}
\eeq
which is exactly the last constraint of (\ref{kerfyugwerluh}). Using this translation into linear programs (\ref{erofuierhopiu}) and (\ref{guiohpui}), it is proved in~\cite{Pryce01}, Thm 3.6, that, if a solution exists to Problem~\ref{liftuerhpituhepu8}, then a unique elementwise smallest solution exists. 
Based on the above analysis, the following \rref{alg:sigmamethod} ($\atomicact{FindOffsets}$) was proposed in~\cite{Pryce01} for solving (\ref{erofuierhopiu},\ref{guiohpui}) and was proved to provide the smallest (real) solution for (\ref{guiohpui}), which happens to be integer.
\begin{algorithm}[ht]
\caption{$\atomicact{FindOffsets}$ ($\cG$ is a bipartite graph with weights $\{d_{\!f\!x}|(f,x)\in\Edges{\cG}\}$)}
	\label{alg:sigmamethod} 
\begin{enumerate}
	\item 
\label{elriuu}	Solve LP (\ref{erofuierhopiu}), which gives a complete matching $\cM$ of maximum weight for $\cG$; any method for solving LP can be used.
	\item \label{elrgfuehliu} Apply the following iteration until a fixpoint is reached (in finitely many steps), from the initial values $c_f=0$:
\begin{enumerate}
	\item \label{erlfiuhpui} $\forall x:d_x\la\max\{d_{\!f\!x}+c_f\mid(f,x)\in\Edges{\cG}\}$~;
	\item \label{louighgpuig} $\forall f:\, c_f\,\la d_x-d_{\!f\!x}$ where  $(f,x)\in\cM$~.
\end{enumerate}
\end{enumerate}
\end{algorithm}

The reason for using the special iterative algorithm for solving the dual LP (\ref{guiohpui}) is that a standard LP-solving algorithm will return an arbitrary solution, not necessarily the smallest one. The following lemma, whic is repeatedly used in Section~11.1 of~\cite{DBLP:journals/corr/abs-2008-05166} about the structural analysis of multimode DAE systems, is an obvious consequence of the linearity of problems (\ref{erofuierhopiu}) and (\ref{guiohpui}):
\begin{lemma}
	\label{rilugtlo}  Let $\cG$ be a given bipartite graph and let two families of weights 
	$(d^1_{\!f\!x})_{(f,x)\in\Edges{\cG}}$ and 	$(d^2_{\!f\!x})_{(f,x)\in\Edges{\cG}}$ be related by 
	$d^2_{\!f\!x}=M\times{d^1_{\!f\!x}}$ for every ${(f,x)\in\Edges{\cG}}$, where $M$ is a fixed positive integer. Then, the offsets of the corresponding \sigmamethod\ are also related in the same way: $d^2_x=M\times{d^1_x}$ for every variable $x$, and $c^2_f=M\times{c^1_f}$ for every function $f$.
\end{lemma}
A necessary and sufficient condition for Problem~\ref{liftuerhpituhepu8} to have a solution is that the set of complete matchings for $\cG$ is non-empty. This criterion can be evaluated prior to computing the optimal offsets.
This completes our background material on structural analysis.

\subsection{{The \sigmamethod\ for non-square systems}}
\label{serukygftou}
Structural analysis must be extended to non-square systems if we wish to allow for guards being evaluated at the current instant---in this case, enabled/disabled equations are progressively determined along with the progressive evaluation of guards, see Appendix~A of~\cite{DBLP:journals/corr/abs-2008-05166}. 
Structural analysis is also a powerful tool for addressing  problems other than just the simulation of DAE systems, e.g., the generation of failure indicators for system diagnosis. For both uses, we will need to address other cases than square DAE systems. We develop here an adaptation of the \sigmamethod\ to non-square DAE systems. In our development, we could simply copy the original constructions, lemmas, and proofs, of~\cite{Pryce01}, while performing marginal adjustments to handle non-square systems. This is presented next.

If $F=0$ is a non-square DAE system, then its weighted bipartite graph $\cG$ has equation nodes and variable nodes in different numbers, hence the concept of complete matching does not apply. In this case, we consider instead the following weaker notion:
\beq
\mbox{
	we say that a matching $\cM$ is \emph{equation-complete} if it covers all the equation nodes.
} 
\label{petoghpeo}
\eeq
No equation-complete matching exists for an overconstrained system $F=0$, i.e., a system with less variables than equations.
This being said, we still consider the primal/dual problems~(\ref{erofuierhopiu}) and~(\ref{guiohpui}). Again, if $\cG$ possesses an equation-complete matching, then problem (\ref{erofuierhopiu}) has a solution, since we can complete this matching to find a feasible solution. Hence, the existence of at least one matching is still used, except that this matching has to be equation-complete instead of complete.

Next, consider respective optima of~(\ref{erofuierhopiu}) and~(\ref{guiohpui}). They still satisfy the slackness conditions~(\ref{eogiheoiwy}), which yields the following results.
\begin{itemize}
	\item An optimal solution of~(\ref{erofuierhopiu}) defines a subgraph $\cH\subseteq\cG$ by: $(f,x)\in\cH$ if and only if $\xi_{f\!x}=1$. Note that $\cH$ is not a matching, as one equation may be associated to more than one variable in this graph. In contrast, each variable is associated to exactly one equation.
	\item By the slackness conditions~(\ref{eogiheoiwy}), for each equation $f=0$ and each $x$ such that $(f,x)\in\cH$, $x$ occurs in $f$ with differentiation degree equal to the maximum differentiation degree in the entire system, namely $d_x$.
\end{itemize}
Call $F_\Sigma$ the DAE system defined by $\cH$ and denote by 
$Z=\{x^{d_x}\mid{x{\in}X}\}$ the set of its leading variables. $F_\Sigma$ is then regarded as an algebraic system of equations with $Z$ as dependent variables, and we apply the Dulmage-Mendelsohn (DM) decomposition to it. This yields $(\block_\underapprox,\block_\squared,\emptyset)$, as we know that the overconstrained block $\block_\overapprox$ is empty. 
%The resulting algorithm is given by~\rref{alg:nsqsigmamethod}.
%%
%%%%%%%%%%%%%%%%%%%%%%%%%%%%%%%%%%%%%%%%%%%%%%%%%%%%%%%
%\begin{algorithm}[ht]
%\caption{$\atomicact{nonsquareSigmaMethod}$}\label{alg:nsqsigmamethod}
%\begin{algorithmic}[1]
   %\Require $F$; 
   %\Return $(b,\bar{F}_{{\rm nsq}\Sigma},F_{{\rm nsq}\Sigma})$
	%%Boolean flag $\mathit{success}$ and, if $\mathit{success}=\ttt$, then regular system $F_{{\rm nsq}\Sigma}$ for system of index $0$ and system $\bar{F}_{{\rm nsq}\Sigma}$ for consistent initialization.
   %\State $\success\gets\atomicact{FindMatching}(\cG)$
   %\If{$b$}
   %\State $({F}_\Sigma,\bar{F}_\Sigma)\gets\atomicact{SigmaMethod}\,(F)$ \label{op:louhgutdr}
   %\State $(\block_\underapprox,\block_\squared,\emptyset)\gets \mathrm{DM}({F}_\Sigma)$  \label{op:eriughkiyu}
   %\State $F_{{\rm nsq}\Sigma}\gets\block_\squared$ \label{op:perouihpefor}
   %\State $\bar{F}_{{\rm nsq}\Sigma}\gets$ restrict $\bar{F}_\Sigma$ to reach$(F_{{\rm nsq}\Sigma})$ \label{op:eiugheliuh}
   %\EndIf
   %\end{algorithmic}
%\end{algorithm}
%
%This algorithm requires a (possibly non-square) DAE system $F$. It returns a Boolean flag $b$ and, if $b=\ttt$, then regular system $F_{{\rm nsq}\Sigma}$ for system of index $0$ and system $\bar{F}_{{\rm nsq}\Sigma}$ for consistent initialization. $\cG$ is the bipartite graph of $F$.

\subsubsection{Justification of the extension}
To justify our above adaptation of Pryce's \sigmamethod, the following questions must be answered:
\begin{question} \rm
	\label {oerifueriofu} How to adapt Step~\ref{elrgfuehliu} of \rref{alg:sigmamethod}, that is, the fixpoint iteration used for computing the offsets?
\end{question}
\begin{question} \rm
	\label {eorifuoi} Is existence and uniqueness of the smallest solution still guaranteed in the non-square case?
\end{question}
%
%\albert{After all, the real outcome of \rref{alg:nsqsigmamethod} is $F_{{\rm nsq}\Sigma}$. The offsets found by \rref{alg:sigmamethod} are not of interest to us by themselves. Thus, the following additional question arises:
%\begin{question} \rm
	%\label {eroiueroi} In which sense is $F_{{\rm nsq}\Sigma}$ (returned by \rref{alg:nsqsigmamethod}) the ``best information'' we can expect from system $F$?
%\end{question}}
%We now address Questions~\ref{oerifueriofu} and~\ref{eorifuoi}, related to the \sigmamethod.
%
\myparagraph{Addressing Question~$\ref{oerifueriofu}$} Step~\ref{erlfiuhpui} of \rref{alg:sigmamethod} is not modified. Step~\ref{louighgpuig}, on the other hand, cannot be kept as is, since $\cH$ is no longer a matching and there may be more than one $x$ such that $(f,x){\in}\cH$. Step~\ref{elrgfuehliu} of \rref{alg:sigmamethod} is thus modified as follows:
\begin{enumerate}
	\item[] 
	\begin{enumerate}
		\item \label{leriuhtpieu}
	$\forall x:d_x \gets \max\{d_{f\!x}+c_f\mid(f,x)\in\Edges{\cG}\}$;
 \item 
\label{sweletgihoi}
$\forall f:c_f \gets \max\{d_x-d_{f\!x}\mid(f,x)\in\cH\}$.
\end{enumerate}
\end{enumerate}
We call \rref{alg:sigmamethod}\remph{'} the resulting algorithm.

\myparagraph{Addressing Question~$\ref{eorifuoi}$}  This question is answered by the following lemma:
\begin{lemma}[non-square \sigmamethod]
	\label{ligfu}  \emph{\rref{alg:sigmamethod}\remph{'}} converges to a fixpoint if and only if $\cH$ is an optimal solution of the primal problem $(\ref{erofuierhopiu})$. Let $\bfc^*,\bfd^*$ be this fixpoint when it exists. Then $\bfc^*,\bfd^*$ is an integer-valued optimal solution for the dual problem $(\ref{guiohpui})$ and, for any optimal solution $\bfc^\dagger,\bfd^\dagger$ for this dual problem, $\bfc^\dagger\geq\bfc^*$ and $\bfd^\dagger\geq\bfd^*$ hold.
\end{lemma}
The last statement expresses the uniqueness of the smallest optimal solution.
Lemma~\ref{ligfu} was stated and proved in~\cite{Pryce01}, albeit for square DAE systems. Here, we will be carefully following the steps of that proof, while checking that they suit the non-square case.
\begin{proof}
Following~\cite{Pryce01}, let $$\phi:\bfc\mapsto\phi(\bfc)$$ be the mapping corresponding to the successive application of steps~(a) and~(b) to $\bfc$.

Assume that $\cH$ is an optimal solution of~(\ref{erofuierhopiu}) and $\bfc,\bfd$ is an optimal solution of its dual problem~(\ref{guiohpui}). By slackness conditions~(\ref{eogiheoiwy}), one has $d_x=c_f+d_{f\!x}$ for $(f,x)\in\cH$. 
Starting from $\bfc$, step~(a) yields
\[\bea{rcl}
\bfd(\bfc) &=&\left\{\left.
\max_{(f,x)\in\Edges{\cG}}\,(d_{f\!x}+c_f)\,\right| x{\in}X
\right\}
\\ [1mm]
&=& \left\{\left.
d_{f\!x}+c_f \mbox{ where } (f,x)\in\cH \,\right| x{\in}X
\right\}
\eea
\]
so that, in turn, applying step~(b) yields $\bfc$: hence, $\phi(\bfc)=\bfc$ (actually, in step~(b), all the $d_x-d_{f\!x}$ for $(f,x)\in\cH$ are equal). In other words, $\bfc$ is a fixpoint of $\phi$.

Conversely, let $\cH\subseteq\cG$ be any feasible solution of (\ref{erofuierhopiu}), not necessarily optimal, and let $\bfc^*,\bfd^*$ be any fixpoint for $\phi$ (note that $\phi$ depends on $\cH$).
We have $d^*_x\geq d_{f\!x}+c^*_f$ for any $(f,x)\in\Edges{\cG}$, with equality if $(f,x)\in\cH$. Assume that $\bfc^*\geq{0}$ also holds; as a matter of fact, $\bfc^*,\bfd^*$ is a feasible solution of the dual problem~(\ref{guiohpui}). Finally, $\cH$ is an optimal solution of the primal problem, and $\bfc^*,\bfd^*$ an optimal solution of the dual problem, because they satisfy the slackness conditions~(\ref{eogiheoiwy}), which are known to characterize optimal solutions. By contraposition, if $\cH$ is not optimal, then no fixpoint of $\phi$ exists.

It remains to be proved that $\bfc^*\geq{0}$ holds. By construction, $\phi$ is nondecreasing: $\bfc'\leq\bfc$ implies $\phi(\bfc')\leq\phi(\bfc)$. Denote by $\bfc^0,\bfd^0,\bfc^1,\bfd^1 \dots$ the alternating sequence of $\bfc$'s and $\bfd$'s obtained by looping over steps (a) and (b), alternatively. We have $\bfc^k=\phi(\bfc^{k-1})$ for $k>0$. If $\bfc^0=0$, then $\bfd^0\geq{0}$ and $c^1_f=d^0_x-d_{f\!x}\geq{0}$, where $x$ is any variable such that $(f,x)\in\cH$ (the value $d^0_x-d_{f\!x}$ is independent from this choice). As a consequence, $\bfc_1\geq{0}$. Since $\phi$ is nondecreasing, applying $\phi$ to both sides of $\bfc_1\geq\bfc_0=0$ yields $\bfc_2\geq\bfc_1$, and so on, so that the algorithm yields a nondecreasing sequence, showing that $\bfc^*\geq{0}$.

Finally, if $\bfc^\dagger,\bfd^\dagger$ is another dual-optimal solution for the offsets, then $\bfc^\dagger$ is a fixpoint of $\phi$ and thus $\bfc^\dagger=\phi(\bfc^\dagger)\geq\phi(0)$, hence 
$\bfc^\dagger\geq\phi^k(0)$ for every $k$, whence $\bfc^\dagger\geq\bfc^*$ follows.  
\end{proof}
\subsubsection{Link with the Pantelides algorithm}
In this section, we closely follow Section~5.3 of~\cite{Pryce01}, regarding the comparison of the \sigmamethod\ with Pantelides' method~\cite{pantelides} for the square case, and adapt it to the non-square case.

Pantelides considers DAE systems with degree at most $1$ for all variables. His construction finds a \emph{minimally structurally singular} (MSS) subset of equations and differentiates each equation in this set, creating an augmented system. This is repeated until no more MSS are found. Pantelides writes the system as $n+m$ equations $F(X,Z)=0$, 
where $X$ stands for the $n$ variables whose derivatives appear, $Y$ stands for the $m$ algebraic variables, and $Z=(\dot{X},Y)$. Pantelides' construction applies to this reformulation. For our comparison, the $Z$ coincides with the collection of \emph{leading variables} $x^{d_x}$ where $d_x=\max_fd_{f\!x}$ is the leading differentiation degree of variable $x$ in the system. 

%\clearpage
We assume that the primal AP~(\ref{erofuierhopiu}) has an optimal solution $\cH^*$, and that there exist smallest optimal offsets $c^*_f,d^*_x$ satisfying $\sum_xd^*_x-\sum_fc^*_f=w(\cH^*)$ (the total weight of $\cH^*$, defined as $\sum_{(f,x)\in\cH^*}d_{f\!x}$). Let
\beq
\widehat{d}_x \eqdef \max_fd_{f\!x}
\label{eriuehyopiu}
\eeq
and define the \emph{leading derivative pattern} $L$ of the weighted bipartite graph $\cG$ (with weights $d_{f\!x}$) as the set
\[
L=\{(f,x)\mid d_{f\!x}=\widehat{d}_x\}
\]
For some equation $f$, respectively some subset $E$ of equations, consider
\[
L_f=\{x\mid(f,x)\in{L}\}~\mbox{, and }~L(E)=\bigcup_{f\in{E}}L_f\,.
\]
Following Pantelides, we say that $E$ is \emph{structurally singular (SS)} if $|L(E)|<|E|$. $E$ is \emph{minimally SS (MSS)} if it is SS and has no proper SS subset.
The link with the $\Sigma$-method is established in the following results.
\begin{lemma}
	\label {0589thgui} \
\begin{enumerate}
	\item \label{elriuliu} If $E$ is MSS with $k$ elements, then $|L(E)|=k-1$
	\item The following three statements are equivalent:
\begin{enumerate}
	\item \label{erligfueghli} $L$ has an equation-complete matching, see $(\ref{petoghpeo})$;
	%\footnote{Recall that a matching is \emph{equation-complete} if it covers all the equation nodes, see (\ref{petoghpeo}).}
	\item \label{poeiheruopo} There are no SS subsets;
	\item \label{prohsfo} All $c_i$ are equal to $0$.
\end{enumerate}
\end{enumerate}
\end{lemma}
\begin{proof} 
For statement \ref{elriuliu}, if $L(E)$ has fewer than $k-1$ elements, then removing any element from $E$ yields a set $E'$ satisfying $|L(E')|\leq|L(E)|<k-1=|E'|$. Therefore, $E'$ is SS, which contradicts the minimality of $E$.

The equivalence \ref{erligfueghli}$\Leftrightarrow$\ref{poeiheruopo} is just Hall's Theorem applied to the sets $L_e$. 

\ref{erligfueghli}$\Ra$\ref{prohsfo}: If $L$ has an equation-complete matching $\cH$, then (by definition of $L$) for $(f,x){\in}\cH$, $d_{f\!x}=\widehat{d}_x$. In other words, $d_{f\!x}$ is maximal among all the $d_{\widehat{f}x}$ for $\widehat{f}$ ranging over the set of all equations. This means that $\cH$ is optimal for the primal problem. We have $d^*_x-c^*_f=d_{f\!x}=\widehat{d}_x$ for $(f,x)\in\cH$ hence $c^*_f=0,d^*_x=\widehat{d}_x$ is the (unique) smallest optimal solution for the dual problem.

$\neg$\ref{erligfueghli}$\Ra\neg$\ref{prohsfo}: If \ref{erligfueghli} fails to hold, then the optimal solution $\cH^*$ for the primal problem contains a $(f,x)$ with $d_{f\!x}<\widehat{d}_x$. Then, by the slackness condition~(\ref{eogiheoiwy}):
\beq
d^*_x=d_{f\!x}+c^*_f<\widehat{d}_x+c^*_f\,.
\label{lreifugflwyug}
\eeq
On the other hand, since $d^*_y,c^*_f$ is a dual-optimal solution, we have $d^*_y\geq{d_{fy}}+c^*_f\geq{d_{fy}}$ for all $y$, whence 
\beq
d^*_y\geq\widehat{d}_y\,. \label{lriuhwui}
\eeq
In particular, $d^*_y\geq\widehat{d}_y$; combining this result with~(\ref{lreifugflwyug}) yields $c_f>d^*_x-\widehat{d}_x\geq{0}$, so that~\ref{prohsfo} is false.
\end{proof}
\begin{lemma}
	\label {reiuheiopu} Let $\cH^*$ be a primal-optimal solution as before. If $c^*_{f_o}=0$ is satisfied for some equation $f_o$, then
\begin{enumerate}
	\item \label{eroiughyu} Let $x_o$ be such that $(f_o,x_o){\in}{\cH^*}$, then $(f_o,x_o)\in{L}$;
	\item \label{toiguehjoih} $c_f=0$ holds for any $f\neq{f_o}$ such that $(f,x_o)\in{L}$.
\end{enumerate}
\end{lemma}
\begin{proof}
	 We have 
	 \[\bea{rcl}
	 d_{f_ox_o} &=& d^*_{x_o} \mbox{ (since $d^*$ is optimal and $c^*_{f_o}=0$)}
	 \\ 
	 &\geq& \widehat{d}_{x_o} \mbox{ (by (\ref{lriuhwui}))} \\ 
	 &\geq& d_{f_ox_o} \mbox{ (by (\ref{eriuehyopiu}))}
	 \eea
	 \]
hence we get
\beq
d^*_{x_o} = \widehat{d}_{x_o} \label{leriuggukh}
\eeq
and $d_{f_ox_o}=\widehat{d}_{x_o}$, which proves statement \ref{eroiughyu}. For statement \ref{toiguehjoih}, if $(f,x_o)\in{L}$ for some $f$, then
\[
0=\widehat{d}_{x_o}-d_{fx_o}=d^*_{x_o}-d_{fx_o}\geq{c^*_f}\geq{0}
\]
whence $c^*_f=0$.
\end{proof}
\begin{lemma}
	\label {2g4huaepvifu} For $E$ any MSS, we have $\forall{f}\in{E}:c_f>0$.
\end{lemma}
\begin{proof}
Suppose the lemma is false, meaning that $E_0\eqdef\{f{\in}{E}\mid{c_f}=0\}$ is nonempty. Let $|E|=k$ and $|E_0|=l$, so that $1{\leq}l{\leq}k$. By statement~\ref{eroiughyu} of Lemma~\ref{reiuheiopu}, we have $L(E_0)\geq{l}$. Since $|L(E)|<|E|$ by the assumption that $E$ is an MSS, we deduce $l<k$, hence $E\setminus{E_0}$ is nonempty. Pick any $f\in{E\setminus{E_0}}$, then $c_{f}>0$. By statement~\ref{toiguehjoih} of Lemma~\ref{reiuheiopu}, $L_{f}$ cannot contain any $x$ that is reachable from $E_0$ via $\cH^*$:
\beq
L(E\setminus{E_0})\,\cap\,\left[\exists{f}.\left((E_0{\times}X)\cap\cH^*\right)\right]
=\emptyset
\label{leghluk}
\eeq
The set on the right-hand side of $\cap$ is contained in $L(E)$ (by statement~\ref{eroiughyu} of Lemma~\ref{reiuheiopu}) and has cardinality at least equal to $|E_0|$. By (\ref{leghluk}), we get $|L(E\setminus{E_0})|\leq|L(E)|-|E_0|<|E\setminus{E_0}|$, hence $E\setminus{E_0}$ is a SS subset strictly contained in $E$, thus contradicting the minimality of $E$.
\end{proof}
Following these results, Pantelides' algorithm can be extended to non-square systems in pretty much the same way the  \sigmamethod\ was extended. This non-square Pantelides' algorithm locates MSS subsets and differentiates them, until there is no such subset left. Note that, when Pantelides' algorithm locates a MSS subset $I$ and differentiates it, this amounts to modifying the offsets found by the \sigmamethod\ in the following way:
\begin{itemize}
	\item the $d_j$ remain unchanged
	\item the $c_i$ decrease by $1$ if $i\in{I}$, remain unchanged otherwise.
\end{itemize}
From this and the preceding results, one gets:
\begin{theorem}[equivalence with Pantelides]
	\label {we4npye}  Pantelides' algorithm gives the same results as the $\Sigma$-method, in that the same offsets $c_i$ are found.
\end{theorem}

\subsection{{Differential and difference arrays}}
\label{lrtioghprtjpi}
Differential arrays were proposed in~\cite{CampbellGear1995} as a tool for finding the latent equations of a DAE system. To $F=0$ a DAE system, we associate its $k$-th \emph{differential array} $\cA_k$ defined, for any $k\in\bN_{\geq 0}$, by
\beq
\cA_k&\eqdef&\left[\bea{r}
F \\ \frac{d}{dt}F \\[1mm] \frac{d^2}{dt^2}F \\ \vdots~~ \\ \frac{d^k}{dt^k}F
\eea\right]
\label{ledriuhrgtoidpor}
\eeq
where $\frac{d}{dt}$ denotes the total time derivative. Thus, for $k$ sufficiently large, array system $\cA_k=0$ should contain all the latent equations that must be added to $F=0$ in order to reduce its index to $0$. More precisely, the index reduced counterpart of $F=0$ is the following system with existential quantifiers:
\beq
\exists W : \cA_k(X,W,Y)=0 \label{letrughntloikhj}
\eeq
where:
\begin{itemize}
	\item $Y$ collects the free variables of system $F=0$;
	\item $X$ collects the dependent variables of system $F=0$;
	\item $W$ collects other supplementary variables involved in array $\cA_k$.
\end{itemize}
\emph{Difference arrays} are defined similarly for discrete-time dAE systems, by replacing the $j$-th time derivatives $\frac{d^j}{dt^j}F$ by the $j$-th forward shift $\ppostset{j}{F}$ defined in (\ref{lrigkuhkjjhvahtr}), when building array (\ref{ledriuhrgtoidpor}):
\beq
\cA_k&\eqdef&\left[\bea{l}
F \\ \postset{F} \\[1mm] \ppostset{2}{F} \\ ~\vdots \\ \ppostset{k}{F}
\eea\right]
\label{liftdrjngbklj}
\eeq
The structural analysis of systems with existential quantifiers, of the form (\ref{elrifuhuio}),
was studied in \rref{sec:operiuhlo}. We can apply it to the structural analysis of arrays (\ref{ledriuhrgtoidpor}) or (\ref{liftdrjngbklj}) to find the latent equations of a DAE system, and, thus, its index. 

The resulting algorithm is much less efficient than Pryce's \sigmamethod\ or Pantelides algorithm. The reason is that, unlike the above methods, the one of \rref{sec:operiuhlo} does not exploit the fact that the rows of the array are successive shifts of the same dAE system.
In turn, \emph{the method of \rref{sec:operiuhlo} extends to time-varying discrete-time systems.} We just need to replace, in (\ref{liftdrjngbklj}), the shifted versions of $F$ by the successive $F(t_0),\dots,F(t_k)$ at successive discrete instants $t_0,\dots,t_k$. This is used in~\cite{BENVENISTE2020,DBLP:journals/corr/abs-2008-05166} to develop a structural analysis of finite cascades of mode changes.

%The basics for the structural analysis of multimode DAE systems are now complete.

\section{Summary of results of \cite{BENVENISTE2020,DBLP:journals/corr/abs-2008-05166} regarding multimode DAE systems}
\label{lgiorthlretoighj}
In this section, we briefly review the results of \cite{BENVENISTE2020,DBLP:journals/corr/abs-2008-05166} regarding multimode DAE systems simulation. We indicate where, in~\cite{BENVENISTE2020,DBLP:journals/corr/abs-2008-05166}, the developments of Sections~\ref{lrieughltiughliu} and~\ref{lfdkuvekrughluh} of this report are used.

\subsection{Multimode DAE and dAE systems definition}
In~\cite{BENVENISTE2020,DBLP:journals/corr/abs-2008-05166}, a mathematical definition for multimode DAE systems is provided, and we also define their discrete-time counterpart, namely multimode dAE systems,
 %(the small ``d'' stands for ``difference''), 
of the respective forms
\beq\bea{rrcl}
\mbox{multimode DAE:}&
\mbox{if } \guard_j(x_i\mbox{'s and derivatives}) &\mbox{then}&
f_j(x_i\mbox{'s and derivatives})=0
\\
\mbox{multimode dAE:}&
\mbox{if } \guard_j(x_i\mbox{'s and shifts}) &\mbox{then}&
f_j(x_i\mbox{'s and shifts})=0\,,
\eea
\label{kuhgflkgfkjh}
\eeq
where $\ppostset{k}{x}$, the $k$-shift of $x$, is defined in (\ref{lrigkuhkjjhvahtr}).

\subsection{Nonstandard semantics}
To allow for a uniform handling of modes and mode change events in the structural analysis of multimode DAE, we use a \emph{nonstandard semantics} for DAEs and multimode DAEs. In this nonstandard semantics, 
\beq
\mbox{derivative $\dot{x}$ is interpreted as its first-order forward Euler scheme $\frac{\postset{x}-x}{\vsmall}$,}
\label{gleriuglhwu}
\eeq
where continuous-time is discretized by using an \emph{infinitesimal} positive time step $\vsmall$ and the forward shift $\postset{x}$ is defined in this discrete-time. Here, as already stated at the beginning of Section~\ref{edjywtefjyt},
\[
\mbox{
\begin{minipage}{11cm}
	 ``infinitesimal'' means ``smaller than any positive real number'', which can be given a formal meaning in \emph{nonstandard analysis}~\cite{Cutland,Robinson}.
\end{minipage}
}
\]
 This mapping allows for a discretization of multimode DAE systems with an infinitesimal error, which yields a faithful discrete-time approximation. Of course, there is no free lunch: the nonstandard semantics is not effective in that there is no computer that can execute it. Still, it is highly useful for the structural analysis, which is a symbolic analysis.
The nonstandard semantics of a multimode DAE systems is a multimode dAE system, in which 
\begin{itemize}
	\item occurrences of increments $\frac{\postset{x}-x}{\vsmall}$ captures the continuous-time dynamics of the considered DAE (for this reason, we call \emph{continuous} the modes in which such dynamics occurs), whereas 
	\item \emph{mode changes} are represented by transitions relating the value $x^-$ of a state just before the change and its value $x^+$ right after the change, encoded in the nonstandard semantics as $x$ and $\postset{x}$.
\end{itemize}

\subsection{{Structural analysis of multimode DAE/dAE systems}}
The structural analyses of both DAE and dAE systems coincide, by replacing the differentiation operator $x\mapsto\dot{x}$ arising in DAEs, by the forward shift operator $x\mapsto\postset{x}$ arising in dAEs. 
Pryce's \sigmamethod\ applies to both DAE and dAE; for the latter, the equation offsets $c_f$ indicate how many times each equation $f$ must be shifted.

\subsubsection{Handling continuous modes}
Since structural analyses mirror each other in continuous and discrete-times through the mapping (\ref{gleriuglhwu}), the structural analysis of continuous modes is performed as usual, with the help of the \sigmamethod\ presented in Section~\ref{lfdkuvekrughluh}. This is summarized in the following
\begin{tool}
	\label{hysejfdigvuy} 
	Standard structural analysis can be performed for each continuous mode, in order to add the needed latent equations. The continuous-time structural analysis (where latent equations are found by differentiation) is therefore used to generate code within continuous modes. 
	%On the other hand, in order to prepare for the generation of restart code at mode changes, the discrete-time structural analysis (where latent equations are found by shifting) based on the nonstandard semantics is used.
	%, yielding the same values for the equation and variable offsets.
\end{tool}

\subsubsection{Handling mode changes}
\label{ujydtfweutyf}

\begin{tool}
	\label{suiytdsfi} 
	In order to prepare for the generation of restart code at mode changes, the discrete-time structural analysis (where latent equations are found by shifting) based on the nonstandard semantics is applied to the continuous  modes before and after the change.
	%, yielding the same values for the equation and variable offsets.
\end{tool}
Basic tools~\ref{hysejfdigvuy} and~\ref{suiytdsfi} rely on J. Pryce's \sigmamethod.
Once latent equations have been added for each continuous mode, we may have a conflict, at mode change events separating two successive continuous modes, between:
\medskip

$\bullet$ the next values of states predicted by the dynamics in the previous mode; and
	
$\bullet$ the consistency equations, generated by the structural analysis of the new mode.

\medskip

\noindent Is this pathological? Not quite. Such a situation can occur even for physically meaningful models. If the new continuous mode has positive index, then, nontrivial consistency equations are generated by the structural analysis, which may be conflicting with the dynamics of the continuous mode before the change. Since one cannot regard this model as being incorrect, we need a policy for handling these conflicts. 

\paragraph{Handling isolated mode changes separating two successive continuous modes:}
In~\cite{BENVENISTE2020,DBLP:journals/corr/abs-2008-05166}, it is proposed to call causality to the rescue: for isolated mode changes, we give priority to the previous dynamics and postpone for a while the conflicting consistency equations generated by the new continuous mode. Here, ``for a while'' means ``for a finite number of infinitesimal time steps'', which amounts to zero time in real duration. 
\begin{tool}
	\label{jysdtcfguijy} Identifying the conflicting consistency equations at mode changes is performed by using the Dulmage-Mendelsohn decomposition (Definition~$\ref{elrftuierhfperui}$ and associated Lemmas~$\ref{erlihpfiu}$ and~$\ref{togiethgouiho}$). 
\end{tool}
Conflicting consistency equations of the new continuous mode are postponed for a finite number of nonstandard instants. The resulting discrete-time system operates for finitely many nonstandard instants, and is the basis for recovering the restart values of the states in the new mode.
This establishes the structural analysis of mode changes separating two successive continuous modes.

\paragraph{Handling finite cascades of mode changes:}
We also handle in~\cite{BENVENISTE2020,DBLP:journals/corr/abs-2008-05166} finite cascades of successive mode changes separating two successive continuous modes. Since the dynamics varies over the successive instants of the cascade, time-invariance no longer holds, and, thus, the \sigmamethod\ does not apply. However, the method of \emph{difference arrays,} closely derived from Campbell and Gear's \emph{differential arrays}~\cite{CampbellGear1995}, can be adapted to time-varying discrete-time systems, and invoked. See Section~\ref{lrtioghprtjpi} and formula (\ref{liftdrjngbklj}) for a short introduction to difference arrays.

For time-invariant dynamics, difference arrays are obtained by stacking, one below the other, the dAE dynamics and its successive shifted versions. For a finite cascade of mode changes leading to a new continuous mode, however, we form the associated difference array by stacking, one below the other, the dynamics that holds at each successive instant of the cascade. More precisely, let $t$ be the current instant in the nonstandard semantics, and let $\vsmall$ be the infinitesimal time step. Then, the time-varying array associated to the cascade of mode changes is
\beq
\cA_k(t) &\eqdef& \left[\bea{l}
F(t) \\ F(t+\vsmall) \\ F(t+2\vsmall) \\ ~\vdots \\ F(t+k\vsmall)
\eea\right]
\label{jsgxcjhcgv}
\eeq
where $F(t)$ is the dynamics at the first event of the cascade, $F(t+\vsmall)$ is the dynamics at the second event of the cascade, and so on until the instant $t+k\vsmall$, which sits within the new continuous mode.
Enough rows should be added to the array, so that the leading variables of the first event of the cascade are uniquely determined as functions of the states before the change. As for differential arrays,  the extra variables brought by the successive shifting must be eliminated by existential quantification---see the comments regarding formula (11) introducing differential arrays in~\cite{CampbellGear1995}.
Finite cascades of mode changes are thus handled by using the results of Section~\ref{sec:operiuhlo} regarding equations with existential quantifiers:
\begin{tool}
	\label{jfhvdgjyg} Finite cascades of mode changes are handled by forming the \emph{difference array} $(\ref{jsgxcjhcgv})$ and analyzing it by using the method of Section~$\ref{sec:operiuhlo}$ regarding equations with existential quantifiers.
\end{tool}

\paragraph{Fixpoint between guards and the equations they control:}
In a multimode DAE system, a logico-numerical fixpoint occurs when a variable $x$ is determined by a different dynamics depending on the value of a guard that itself depends, directly or indirectly, on $x$. In~\cite{BENVENISTE2020,DBLP:journals/corr/abs-2008-05166}, only multimode DAE systems having no logico-numerical fixpoint are supported. In such models, the above guard should rather depend on the left-limit $x^-$ of $x$, where $x^-(t)=\lim_{s\nearrow{t}}x(s)$. Nevertheless, in Appendix~A of~\cite{DBLP:journals/corr/abs-2008-05166}, we propose a method for handling multimode DAE systems with fixpoints but no cascades of length $>1$. This method uses the following basic tool:
\begin{tool}
	\label{ikjuftdsfgiuy} 
The \sigmamethod\ for non-square systems, presented in Section~\ref{serukygftou}.
\end{tool}
%\clearpage
\subsubsection{Generic form of the generated simulation code}
For $S$ a multimode DAE model having continuous modes separated by finite cascades of mode change events, the method proposed in~\cite{DBLP:journals/corr/abs-2008-05166} returns the following:

\paragraph{Case of a structurally incorrect model:} The compilation algorithm returns, for each mode and mode change event in which structural analysis fails, the overdetermined subsets of equations and the underdetermined subsets of variables.

\paragraph{Case of a structurally correct model, producing an \emph{\mDAE\ interpreter}:} The compilation algorithm returns an \emph{\mDAE\ interpreter,} which is a labeled bipartite graph $\cG$ whose vertices are:
	\begin{itemize}
		\item state and leading variables $x$ of the system, as well as guards $\guard$; 
		\item structurally nonsingular blocks $\block$ of algebraic equations determining leading variables from state variables, and expressions $\mathit{exp}$ determining the values of guards as function of state and leading variables.
	\end{itemize}
	A branch $x\ra\block$ exists in $\cG$ if $x$ is a free variable (input) of block $\block$ and a branch $\block\ra{x}$ exists in $\cG$ if $x$ is a dependent variable (output) of block $\block$. A branch $x\ra\mathit{exp}$ exists in $\cG$ if $\mathit{exp}$ has $x$ as one of its arguments, and a branch $\mathit{exp}\ra{\guard}$ exists in $\cG$ if $\guard$ is computed using expression $\mathit{exp}$.
	
Not every block is involved in a given mode or mode change. Thus, to each block $\block$ we associate a \emph{label} $\lambda(\block)$, which is a property characterizing the set of all modes or mode changes in which $\block$ is active. A label is then assigned to every \emph{path} $\path$ traversing blocks $\block_1,\dots,\block_k,\dots,\block_K$ by setting $\lambda(\path)=\bigwedge_1^K\lambda(\block_k)$. Every circuit of $\cG$ has label $\fff$, which ensures that the code is not globally circular. Also, for any two blocks $\block_1\neq\block_2$ leading to the same variable $x$ in $\cG$, we have $\lambda(\block_1)\wedge\lambda(\block_2)=\fff$, which ensures that two different blocks never compete at determining the same variable $x$.
This interpreter encodes all the different schedulings that are valid for the different modes of the system. Note that modes are not enumerated, which is essential in order for this technique to scale up.

This technique is reminiscent of the \emph{conditional dependency graphs} used in the compilation of the Signal synchronous language~\cite{benveniste-distribution00,synchronous-twelve-years-later}.

The notion of \mDAE\ interpreter is illustrated in~\rref{fig:rldc2cdg}, for an RLDC2 circuit example with schematic shown in \rref{fig:rldc2-sch} and model given in \rref{sys:rldc2-system}, see Section 12 of~\cite{DBLP:journals/corr/abs-2008-05166} for details. This form of \mdAE\ interpreter is used in the \IsamDAE\ tool~\cite{Caillaud2020a}, by representing all the label using BDD-related data structures.\footnote{BDD means: Binary Decision Diagram, a data structure developed for model checking.}
\begin{figure}[!h]
\begin{center}
\includegraphics[width=0.80\linewidth]{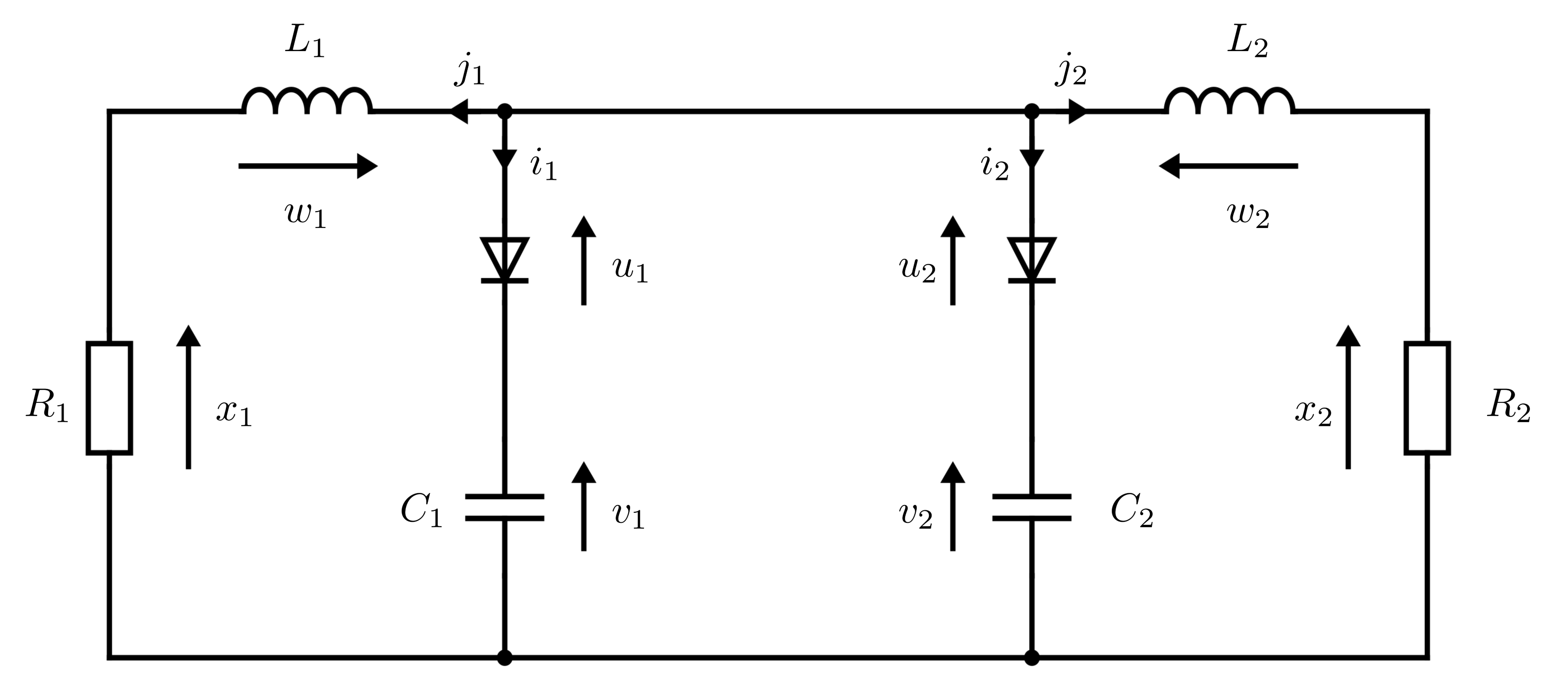}
\end{center}
\caption{Schematics of the RLDC2 circuit.}
\label{fig:rldc2-sch}
\end{figure}
\begin{figure}[!ht]
  \begin{equation}
    \begin{array}{rclc}
      0 &=& i_1 + i_2 + j_1 + j_2 & (K_1) \\
      x_1 + w_1 &=& u_1 + v_1 & (K_2) \\
      u_1 + v_1 &=& u_2 + v_2 & (K_3) \\
      u_2 + v_2 &=& x_2 + w_2 & (K_4) \\
      w_1 &=& L_1 \cdot j'_1 & (L_1) \\
      w_2 &=& L_2 \cdot j'_2 & (L_2) \\
      i_1 &=& C_1 \cdot v'_1 & (C_1) \\
      i_2 &=& C_2 \cdot v'_2 & (C_2) \\
      x_1 &=& R_1 \cdot j_1 & (R_1) \\
      x_2 &=& R_2 \cdot j_2 & (R_2) \\
      s_1 &=& \text{if } \guard_1 \text{ then } i_1 \text{ else } -u_1 & (S_1) \\
      s_2 &=& \text{if } \guard_2 \text{ then } i_2 \text{ else } -u_2 & (S_2) \\
      0 &=& \text{if } \guard_1 \text{ then } u_1 \text{ else } i_1 & (Z_1) \\
      0 &=& \text{if } \guard_2 \text{ then } u_2 \text{ else } i_2 & (Z_2) \\
      \guard_1 &=& (s_1^- \geq 0) & (\guard_1^-) \\
      \guard_2 &=& (s_2^- \geq 0) & (\guard_2^-)
    \end{array}
    \label{sys:rldc2-system}
  \end{equation}
  \caption{The RLDC2 model (\emph{n.b.}: variable $y'$ denotes the time derivative of $y$).}\label{fig:rldc2}
\end{figure}
\begin{figure*}[!htp]
    \begin{center}
    \includegraphics[width=0.8\linewidth]{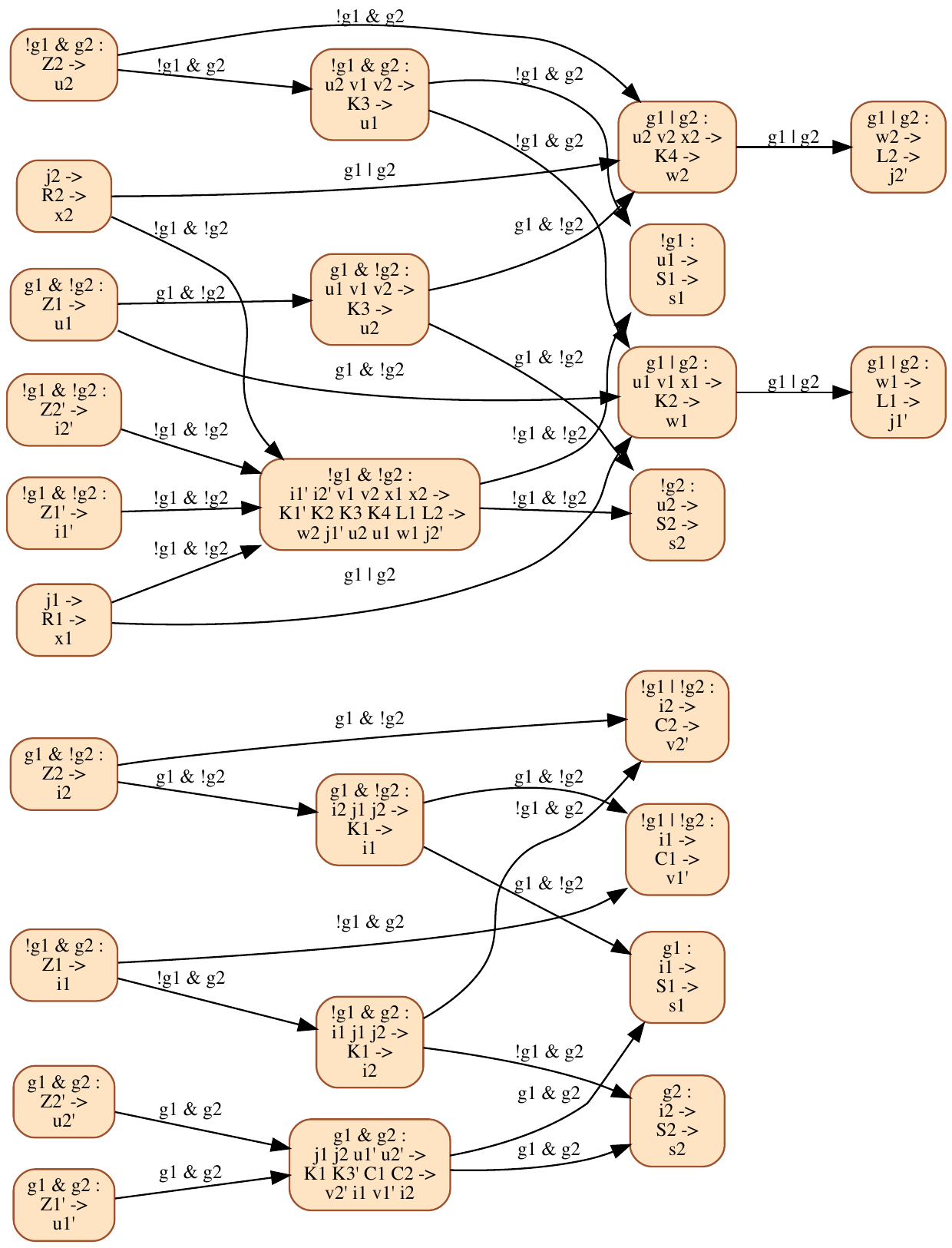}
    \end{center}
    \caption{Block dependency graph of the RLDC2 model, generated
    by IsamDAE. Vertices are labeled
    $p : Y \rightarrow \block \rightarrow X$, where: $p$
    is a propositional formula defining in which modes the block is
    evaluated; $Y$ is the set of free variables for reading; $\block$ is the set of
    equations of the block; $X$ is the set of dependent variables for writing. Edges
    are labeled by a propositional formula, defining in which modes the
    dependency applies---notation ``!g'' means ``not g'', ``g1 \& g2'' is the conjunction, and ``g1 $|$ g2'' is the disjunction.}\label{fig:rldc2cdg}
  \end{figure*}

\section{Conclusion}
This report complements references~\cite{BENVENISTE2020,DBLP:journals/corr/abs-2008-05166} by collecting and providing all missing details related to structural analysis, for both algebraic and DAE systems, in a systematic way. 
\clearpage
\bibliographystyle{abbrv}
\bibliography{hybrid,modelica,synchrone,lucy}

\begin{thebibliography}{10}

\bibitem{benveniste-distribution00}
A.~Benveniste, B.~Caillaud, and P.~L. Guernic.
\newblock Compositionality in dataflow synchronous languages: specification \&
  distributed code generation.
\newblock {\em Information and Computatin}, 163:125--171, 2000.

\bibitem{BENVENISTE2020}
A.~Benveniste, B.~Caillaud, and M.~Malandain.
\newblock The mathematical foundations of physical systems modeling languages.
\newblock {\em Annual Reviews in Control}, 2020.

\bibitem{DBLP:journals/corr/abs-2008-05166}
A.~Benveniste, B.~Caillaud, and M.~Malandain.
\newblock The mathematical foundations of physical systems modeling languages.
\newblock {\em CoRR}, abs/2008.05166, 2020.

\bibitem{synchronous-twelve-years-later}
A.~Benveniste, P.~Caspi, S.~Edwards, N.~Halbwachs, P.~{Le Guernic}, and R.~{de
  Simone}.
\newblock The synchronous languages 12 years later.
\newblock {\em Proceedings of the IEEE}, 91(1), Jan. 2003.

\bibitem{Berge1962}
C.~Berge.
\newblock {\em {The theory of graphs and its applications}}.
\newblock Wiley, 1962.

\bibitem{Caillaud2020a}
B.~Caillaud, M.~Malandain, and J.~Thibault.
\newblock Implicit structural analysis of multimode {DAE} systems.
\newblock In {\em 23rd ACM International Conference on Hybrid Systems:
  Computation and Control (HSCC 2020)}, Sydney, Australia, April 2020.
\newblock to appear.

\bibitem{CampbellGear1995}
S.~L. Campbell and C.~W. Gear.
\newblock {The index of general nonlinear DAEs}.
\newblock {\em {Numer. Math.}}, 72:173--196, 1995.

\bibitem{Cartan}
H.~Cartan.
\newblock {\em {Formes Diff\'erentielles}}.
\newblock Collection M\'ethodes. {Hermann}, 1967.

\bibitem{Cutland}
N.~Cutland.
\newblock {\em {Nonstandard analysis and its applications}}.
\newblock Cambridge Univ. Press, 1988.

\bibitem{DieudonneEA1}
J.~Dieudonn\'e.
\newblock {\em {Fondements de l'analyse moderne}}.
\newblock Gauthier-Villars, 1965.

\bibitem{Duff86}
I.~S. Duff, A.~M. Erisman, and J.~K. Reid.
\newblock {\em Direct Methods for Sparse Matrices}.
\newblock Numerical Mathematics and Scientific Computation. Oxford University
  Press, 1986.

\bibitem{Hogb06}
L.~Hogben, editor.
\newblock {\em Handbook of Linear Algebra}.
\newblock CRC Press, Boca Raton, FL, USA, 2006.

\bibitem{hogben14}
L.~Hogben.
\newblock {\em Handbook of Linear Algebra}.
\newblock 2nd edition, 2014.

\bibitem{MattssonSoderlin1993}
S.~E. Mattsson and G.~S\"oderlind.
\newblock {Index reduction in Differential-Algebraic Equations using dummy
  derivatives}.
\newblock {\em {Siam J. Sci. Comput.}}, 14(3):677--692, 1993.

\bibitem{pantelides}
C.~Pantelides.
\newblock {The consistent initialization of differential-algebraic systems}.
\newblock {\em {SIAM J. Sci. Stat. Comput.}}, 9(2):213--231, 1988.

\bibitem{PothenF90}
A.~Pothen and C.~Fan.
\newblock Computing the block triangular form of a sparse matrix.
\newblock {\em {ACM} Trans. Math. Softw.}, 16(4):303--324, 1990.

\bibitem{Pryce01}
J.~D. Pryce.
\newblock {A simple structural analysis method for DAEs}.
\newblock {\em {BIT}}, 41(2):364--394, 2001.

\bibitem{Robinson}
A.~Robinson.
\newblock {\em {Nonstandard Analysis}}.
\newblock Princeton Landmarks in Mathematics, 1996.
\newblock {ISBN} 0-691-04490-2.

\end{thebibliography}

\end{document}